\documentclass[12pt]{amsart}
\usepackage[dvipdfmx]{graphicx}
\usepackage{amsmath,amssymb, amscd}

\makeatletter
\@addtoreset{equation}{section}
\makeatother
\usepackage{enumerate}
\usepackage{color}

\def\Be{\begin{equation}}
\def\Ee{\end{equation}}
\def\Bea{\begin{eqnarray}}
\def\Eea{\end{eqnarray}}

\def\ee{\mbox{e}}
\def\ii{{\sqrt{-1}}}
\def\dd{{\rm{d}}}

\def\RR{{\mathbb R}}

\def\CC{{\mathbb C}}

\def\fB{{\mathfrak B}}

\def\fM{{\mathfrak M}}
\def\fZ{{\mathfrak Z}}
\def\cE{{\mathcal E}}

\def\CC{{\mathbb C}}
\def\RR{{\mathbb R}}

\def\bt{{\bf t}}
\def\bn{{\bf n}}

\def\ee{{\rm{e}}}
\def\mm{{\rm{m}}}
\def\cc{{\rm{c}}}
\def\co{{\rm{co}}}
\def\dor{{\rm{do}}}

\def\Book#1{{\it{#1}, }}
\def\Paper#1{{\textit{#1}, }}
\def\Jour#1{{\rm{#1} }}
\def\Yr#1{({\rm{#1}) }}

\def\Vol#1{{\textbf{#1}}}
\def\Pages#1{{\rm{#1}}}

\def\Publaddr#1{{\rm{#1}, }}
\def\Publ#1{{\rm{#1}, }}
\def\By#1{{\rm{#1}, }}

\newtheorem{definition}{Definition}[section]
\newtheorem{theorem}[definition]{Theorem}
\newtheorem{proposition}[definition]{Proposition}

\newtheorem{remark}[definition]{Remark}
\newtheorem{lemma}[definition]{Lemma}

\begin{document}

\thispagestyle{plain}

\title{On $\Lambda$-Elastica }

\author{S. Matsutani, H. Nishiguchi, K. Higashida, \\ 
A. Nakatani,  H. Hamada}

\date{}

\maketitle

\begin{abstract}
In this paper, we investigate a transition from an elastica to a piece-wised
elastica whose connected point defines the hinge angle $\phi_0$;
we refer the piece-wised
elastica  {\it{$\Lambda_{\phi_0}$-elastica}} or {\it{$\Lambda$-elastica}}.
The transition  appears in the bending beam experiment;
we compress elastic beams gradually and then suddenly
due the rupture, the shapes of $\Lambda$-elastica appear.
We construct a mathematical theory to describe the phenomena 
and represent the $\Lambda$-elastica in terms of
the elliptic $\zeta$-function completely.
Using the mathematical theory,
we discuss the experimental results from an energetic viewpoint
and numerically show the 
explicit shape of $\Lambda$-elastica.
It means that 
this paper provides a novel investigation on elastica theory with rupture.
%
\end{abstract}


\bigskip
\section{Introduction}
\bigskip


The elastica problem is the oldest minimal problem with
the Euler-Bernoulli energy functional \cite{Griffiths,Mat13,Tr}.
In the set of the isometric analytic immersions of
$(s_1, s_2)$ into $\CC$ for fixed $s_1, s_2\in \RR$, $(s_1 < s_2)$,
the minimal point of the energy functional corresponds to the shape of
the elastic curve or elastica.
There are so many studies on the real materials and nonlinear phenomena 
related to elastica, e.g., \cite{Bigoni, MH}.

In this paper, we have investigated the elastic beam which is 
allowed to have 
a transition from the set of isometric analytic immersions to 
the set of continuum immersions which are analytic 
 except a certain point. 
We assume that
the transition occurs depending on its critical force
at the point which has the maximal force in the elastic beam.
Then we have an interesting shape which we call $\Lambda$-elastica.

More precisely,
we consider the set of the isometric analytic immersions,
$\fM_{(s_1,s_2)}:=\{ Z: (s_1,s_2) \to \CC : $ an isometric 
 analytic immersion$\}$.
The  Euler-Bernoulli energy functional is given 
by $\displaystyle{\frac{1}{2}\int_{s_1}^{s_2} k^2 d s}$
where 
$k = \displaystyle{\frac{1}{\ii} \partial_s \log \partial_s Z}$
and $\displaystyle{\partial_s = \frac{d}{d s}}$.
The elastica is given as the minimizer of the energy.
Further for a point $s_0\in (s_1, s_2)$
and a real parameter $\phi_0$,
the transition is
from $\fM_{(s_1,s_2)}$  to 
$\fM^{s_0,\phi_0}_{(s_1,s_2)}:= \Bigr\{ Z: (s_1,s_2) \to \CC \ \Bigr| 
\ $ continues, 
$\phi_0 = \displaystyle{\frac{1}{\ii}  
\log \frac{\partial_s Z(s_0+0)}{\partial_s Z(s_0-0)}}$,
$\rho^{(s_1,s_2)}_{(s_1, s_0)} Z \in \fM_{(s_1,s_0)}$ and 
$\rho^{(s_1,s_2)}_{(s_0, s_2)} Z \in \fM_{(s_0,s_1)} \Bigr\}$ 
for the condition.
Here
$\rho^U_V$ is the restriction operator
which restricts the domain of the function from $U$ to
$V$ ($V\subset U$). 
Corresponding to $\fM^{s_0,\phi_0}_{(s_1,s_2)}$, we consider the 
minimal problem of the energy 
$\displaystyle{\frac{1}{2}\int_{s_1}^{s_0} k^2 d s}+$
$\displaystyle{\frac{1}{2}\int_{s_0}^{s_2} k^2 d s}$.
The minimizer is called $\Lambda$-elastica in this paper.
The parameter 
$\phi_0$ is the angle to determines the shape of 
$\Lambda$-elastica, and thus we, precisely, say 
$\Lambda_{\phi_0}$-elastica.

In this paper, we express deformation of elastic beams as a disjoint
 orbit in a function space which contains $\fM_{(s_1,s_2)}$  and
$\fM^{s_0,\phi_0}_{(s_1,s_2)}$ which describe the transition from 
elastica to $\Lambda_{\phi_0}$-elastica
mathematically.

This work was motivated from the kink phenomena \cite{Hagihara}.
The plastic deformation occurs due to the generations of dislocations
\cite{HessBarrett}.
The plastic deformation causes kink phenomena.
In the kink phenomena, there appear various shapes \cite{Bigoni} and
we find some shapes which could be written by  parts of elastica,
or $\Lambda$-elastica as mentioned above.
In this stage, we do not find a reasonable connection between
the shape of elastica and the kink phenomena.
However it is natural to investigate  $\Lambda$-elastica because
in \cite{DMJ}, the same problem for the thin Kapton
membranes was studied using the finite element method and
 there appeared similar shapes in 
the stretching elastic looped ribbons in \cite{MWT}.
Further it is also interesting to consider the transition 
from elastica to $\Lambda$-elastica as we show experimental results
in this paper.

In order to consider the transition from elastica to $\Lambda$-elastica,
we first show the experimental results of beam bending test with 
rupture phenomena in Section \ref{sec:Exp}.
When the compressed force to the elastic beam is greater than a critical force,
the elastic beam is broken at the critical state 
in which the local force is the maximal
value. Due to the energy of rupture, the total energy of this system
decreases.
There appear $\Lambda$-elastica at the bounce-back of the pieces of the
broken elastic beam after they separate. It apparently behaves like
a continuum beam and we find an 
angle $\phi_0$ and the shape of $\Lambda_{\phi_0}$-elastica. 
We show the compression experiments of elastic
beams of different thickness which correspond to
 different effective elastic constants. Section \ref{sec:elastica}
is a review section of the elastica theory following \cite{Mat13}.
The shape of elastica is  described well in terms of Weierstrass
elliptic $\zeta$-function,
though we do not consider the boundary condition 
explicitly there. In order to explain the experimental results of the beam
bending test,
 we explicitly describe
the boundary condition in the elastica problem in Section \ref{sec:Trans}.
After then, we investigate
the transition from elastica to $\Lambda_{\phi_0}$-elastica with hinge 
$\phi_0$. Section \ref{sec:Trans} is our main part in this paper.
There we construct a mathematical theory to describe
the experimental phenomena 
and represent the $\Lambda$-elastica in terms of
the elliptic $\zeta$-function completely.
In Section \ref{sec:Dis}, we discuss the relation between 
theoretical results and experimental results using the mathematical theory.
It means that we provide a novel investigation on elastica theory with rupture.

\section{Experimental Results}\label{sec:Exp}
\subsection{Experimental Results of Elastica and $\Lambda$-elastica} 
\label{sec:Experiment}

In order to express our motivation in this study, 
we show our experimental results. 
As in Figures \ref{fig:apparatus} and \ref{fig:experiment},
we experimented the beam bending test for
the three type samples of plastic panels as elastic beams
$\delta \times L' \times L$,
where $L'$ is its width, 20.0 [mm], $L$ is its length, 300.0 [mm] and 
$\delta$ is the thickness, 2.0[mm], 3.0[mm] and 5.0[mm].
They consists of the same plastic material and 
the difference of the thickness means the difference of the effective
elastic constant $\kappa \delta$ as mentioned in Section \ref{sec:bend}.
We used a compression testing apparatus,
Autograph AG-100kNG made by Shimadzu Corporation,
in which we can fold the endings 
of the panels so that the ending are parallel and the same 
horizontal position.

\begin{figure}[ht]
\begin{center}
\includegraphics[width=0.35\hsize]{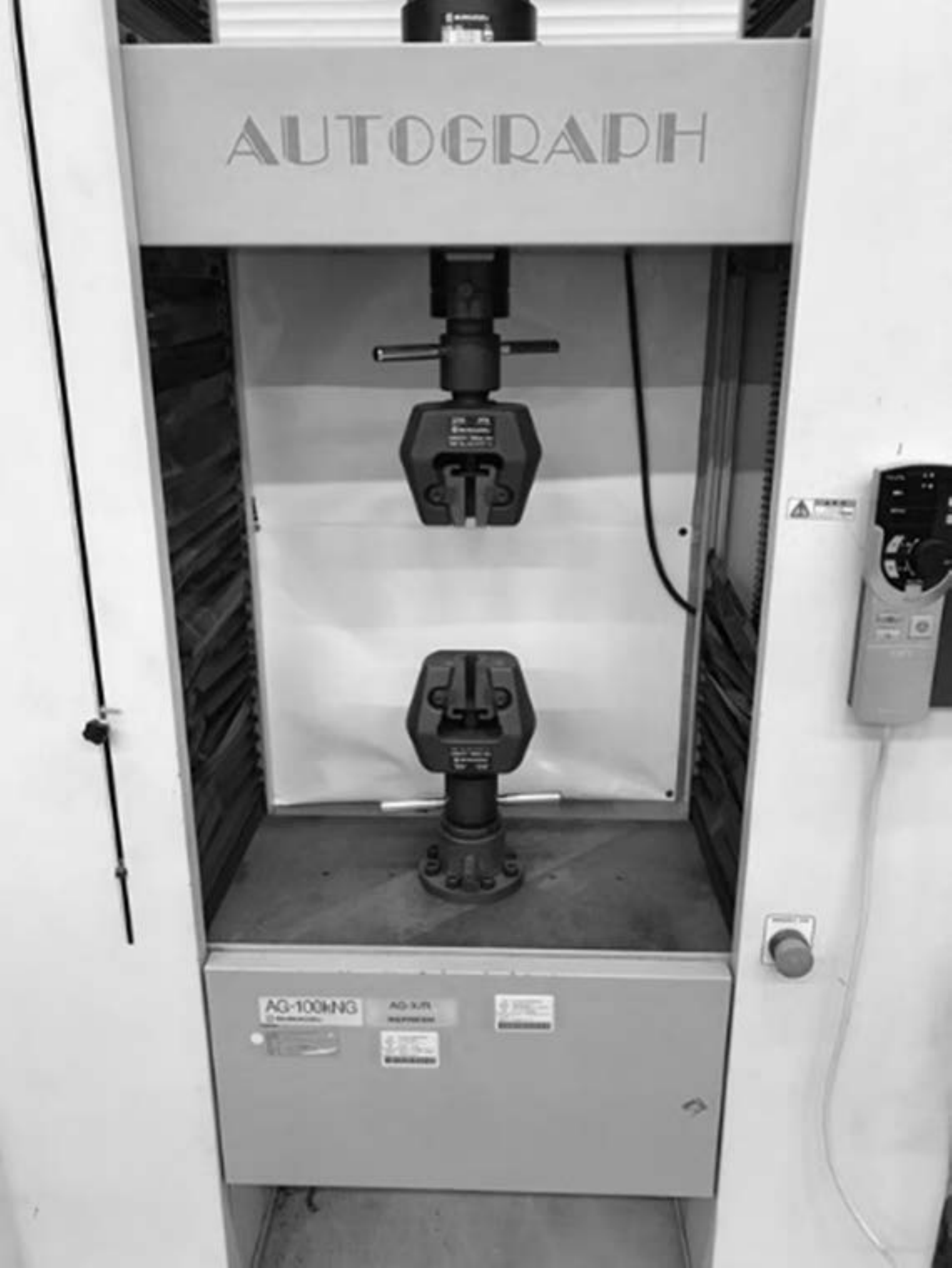}
\end{center}
\caption{Autograph AG-100kNG.}
\label{fig:apparatus}
\end{figure} 

In the experiments,
the crosshead speed was 10[mm/min]. 
The Phantom high-speed camera was used to capture 
the bent panel just before buckling and just after buckling.
The frame rate was 10000[frame/sec]. The length of folded area was 20.0[mm] 
at each ends, therefore the length of bending part was 260.0[mm].
By preserving the parallel and the same levels, we can compress them
to observe the bending structure. 
We gradually compress the panel and then the panel broke suddenly.
We refer the state {\it{critical state}}.
At the state, we denote the height by $X_\cc$, the width by $W_\cc$
and  the curvature by $k_\cc$  as in Figure \ref{fig:ELcurve} (a).
We call $k_\cc$ {\it{critical curvature}},
$X_\cc$ {\it{critical height}} and 
$W_\cc$ {\it{critical width}}..
At the critical states, the shape of the elastic panels are displayed in
Figure \ref{fig:experiment} (a), (b) and (c).
The unit of scale in the background is given as 18.89[mm/unit].
The rupture needs the energy $\Delta E$ and the system lost the energy.
After pieces of the broken elastic panel separate,
the panel satisfies continuous condition at the bounce-back of the pieces
and there appear a hinge which
connects pieces.
In other words, we find the $\Lambda$-elastica as in
Figure \ref{fig:experiment} (d), (e) and (f).
The hinge angle $\phi_0$ and the height $X_\Lambda$
is defined as in Figure \ref{fig:ELcurve} (b).
As we are concerned with the transition from elastica to $\Lambda$-elastica at 
the critical state,
the experimental results can be regarded as the transition.

In order to obtain these shapes, the notch was introduced at the surface of 
testing panel whose depth was 0.5[mm] and width was 1.0[mm] in 
the direction perpendicular to the longitudinal direction.

\begin{figure}[ht]
\begin{minipage}{0.32\hsize}
\begin{center}
\includegraphics[angle=-90, width=0.9\hsize]{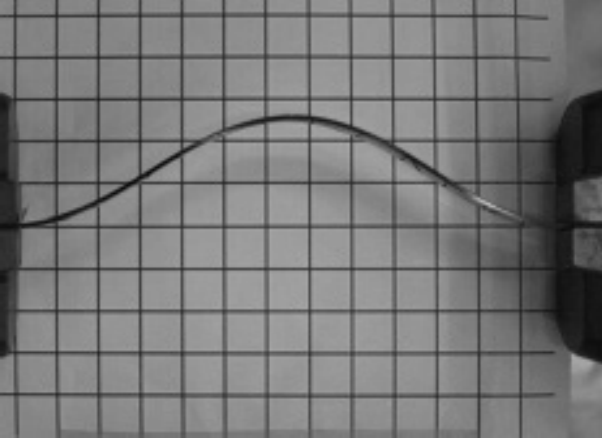}
\\
(a)
\end{center}
\end{minipage}
\begin{minipage}{0.32\hsize}
\begin{center}
\includegraphics[angle=-90, width=0.9\hsize]{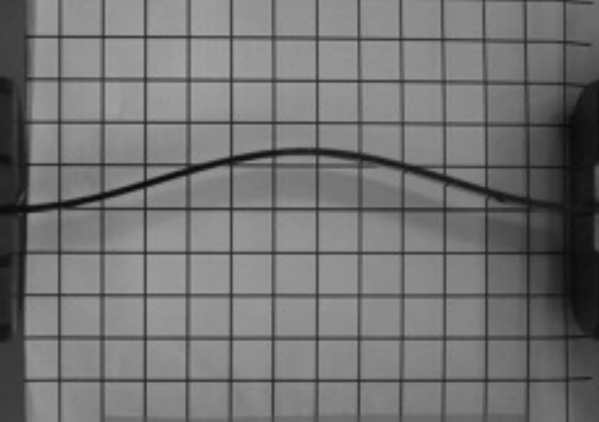}
\\
(b)
\end{center}
\end{minipage}
\begin{minipage}{0.32\hsize}
\begin{center}
\includegraphics[angle=-90, width=0.9\hsize]{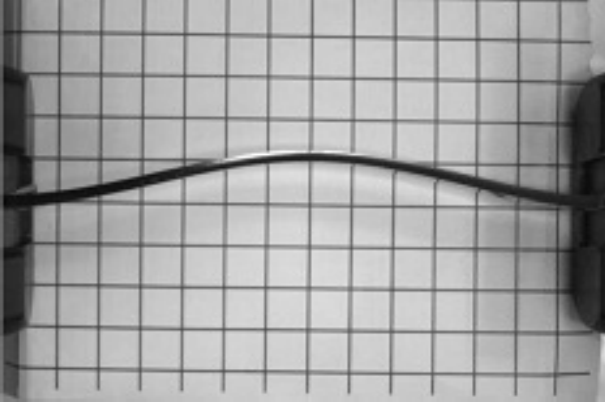}
\\
(c)
\end{center}
\end{minipage}
\\
\begin{minipage}{0.32\hsize}
\begin{center}
\includegraphics[angle=-90, width=0.9\hsize]{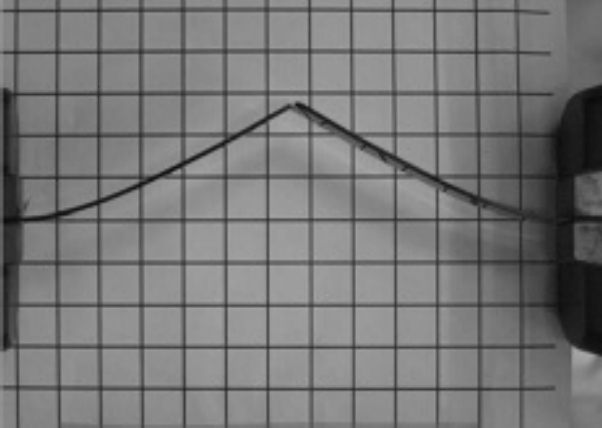}
\\
(d)
\end{center}
\end{minipage}
\begin{minipage}{0.32\hsize}
\begin{center}
\includegraphics[angle=-90, width=0.9\hsize]{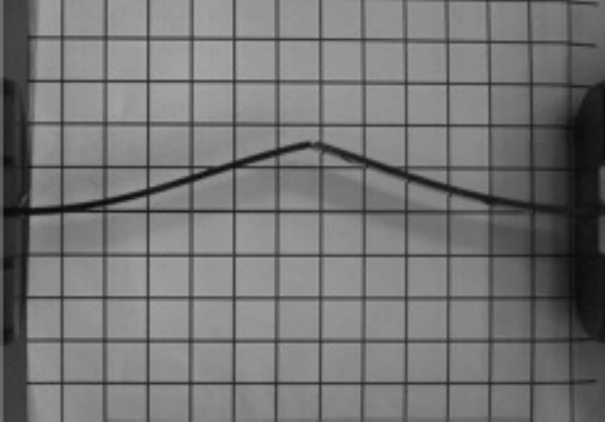}
\\
(e)
\end{center}
\end{minipage}
\begin{minipage}{0.32\hsize}
\begin{center}
\includegraphics[angle=-90, width=0.9\hsize]{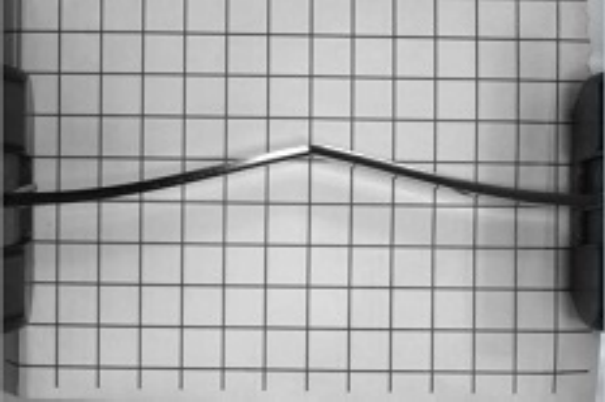}
\\
(f)
\end{center}
\end{minipage}

\caption{Press experiments of elastic beam: 
(a)-(c) are elastic panels with critical curvature
 whereas (d)-(f) are of $\Lambda_{\phi_0}$ 
shapes of elastic panels which appear at the bounce-back of 
the separated pieces of panels and they behave like continuum beams.
The thickness $\delta$ of (a) and (d) are 2.0[mm],
(b) and (e) correspond to 3.0[mm] and (c) and (f) to 5.0[mm].
}
\label{fig:experiment}
\end{figure} 

\begin{figure}[ht]
\begin{center}
\includegraphics[width=0.65\hsize]{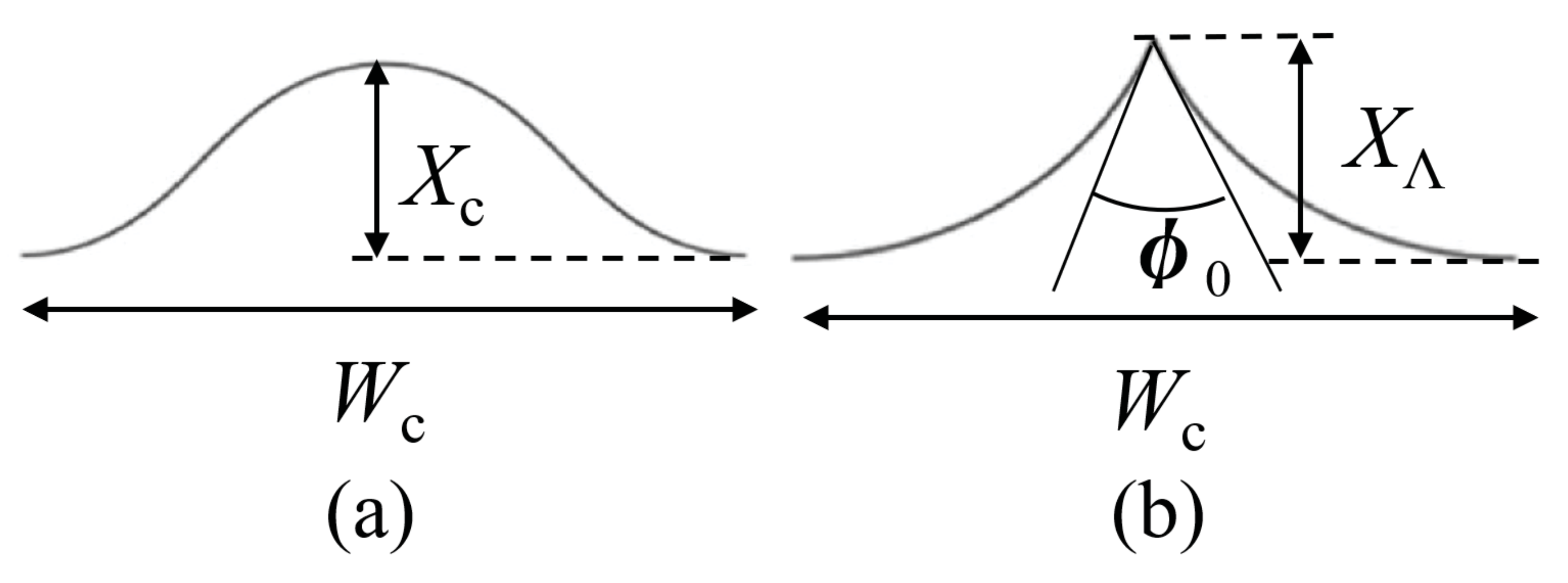}
\end{center}
\caption{Geometrical Characteristics: $X_\cc$, $W_\cc$ and $\phi_0$}
\label{fig:ELcurve}
\end{figure}

Dependence of $X_\cc$, $W_\cc$, $\phi_0$ and $X_\Lambda$
on the thickness $\delta$
is shown in Table 1.
The thicker is, the lager the critical width $W_\cc$ is. 
the smaller, the height $X_\cc$ is and  the larger the angle $\phi_0$ is.
It should be noted that $X_\Lambda$ is nearly equal to $X_\cc$.

\begin{table}[htb]
Table 1: The thickness vs $X_\cc$, $W_\cc$, and $\phi_0$
 in  Figure \ref{fig:experiment} ($W_0=14.1$)
  \begin{tabular}{|c|c|c|c|c|}
\hline
$\delta$ & $X_\cc$ &  $W_\cc$  &$\phi_0$ & $X_\Lambda$ \\
\hline
2.0 [mm]& 49[mm]  & 234[mm]&  0.66$\pi$& 51[mm]\\
\hline
3.0 [mm]& 25[mm]  & 242[mm] & 0.79$\pi$& 28[mm]\\
\hline
5.0 [mm]& 21[mm] & 250[mm] &  0.86$\pi$ & 23[mm]\\
\hline
  \end{tabular}
\end{table}

It is hard to control the transition in this experiment
 but if there is a certain geometrical 
constraint so that it must be continuous
even after it was broken, we may find $\Lambda$-elastica 
statically.
By assuming the situation,
we investigate this experimental result mathematically.

\subsection{Thickness and elastic constant of the elastica}\label{sec:bend}
In order to show the relation between the thickness of elastic beam and
the effective elastic constant,
let us consider an embedding of the elastic beam with constant thickness
$\delta$ in the complex plane $\CC$.
Assume that the center axis of the beam does not change its length.
We estimate the
 stretching of the elastic beam.
Let the curve be
parallel to the center axis curve with
 the vertical distance $q$ from the center axis,
which is parameterized by $s_q$ with the euclidean distance.
The stretching of the curve is given by
$$
d s_q = (1+ k(s)q ) d s,
$$
where  $s$ is the arclength of  the center axis of the beam,
$k(s)$ is the curvature whose inverse is the curvature radius
$\rho(s)=1/k(s)$, and   
$q\in [-\delta/2, \delta/2]$. We assume the case $\delta/\rho =\delta\cdot k \ll 1$.
It means that $e_{q}:= 1+k(s)q=\displaystyle{\frac{\partial s_q}{\partial s}}$ is the ratio of
the stretching length.
The free energy density $\mathcal F d q d s$ caused by bending is given by
$$
\mathcal F d q d s=
\frac{1}{2}\kappa\left(\frac{\partial e_q}{\partial q}\right)^2 (1+k q)
d q d s,
$$
where $\kappa$ is the elastic constant.
By integrating along the vertical direction, we have
\begin{gather}
\begin{split}
\left(\int_{-\delta/2}^{\delta/2} \mathcal F d q \right)
d s
&=\frac{1}{2}\delta
\left(\kappa k^2 +\frac{1}{2}\delta k^3\right)d s\\
&=\frac{1}{2}\delta
\kappa k^2\left(1 +o\left(\frac{\delta}{k}\right)\right)d s.\\
\end{split}
\label{eq:delta_k^2}
\end{gather}
The factor $\kappa \delta$ is regarded as an effective elastic constant,
which  is proportional to the thickness $\delta$.
(\ref{eq:delta_k^2}) is known as the density of 
the Euler-Bernoulli energy functional.

\begin{figure}[ht]
\begin{center}
\includegraphics[width=0.45\hsize]{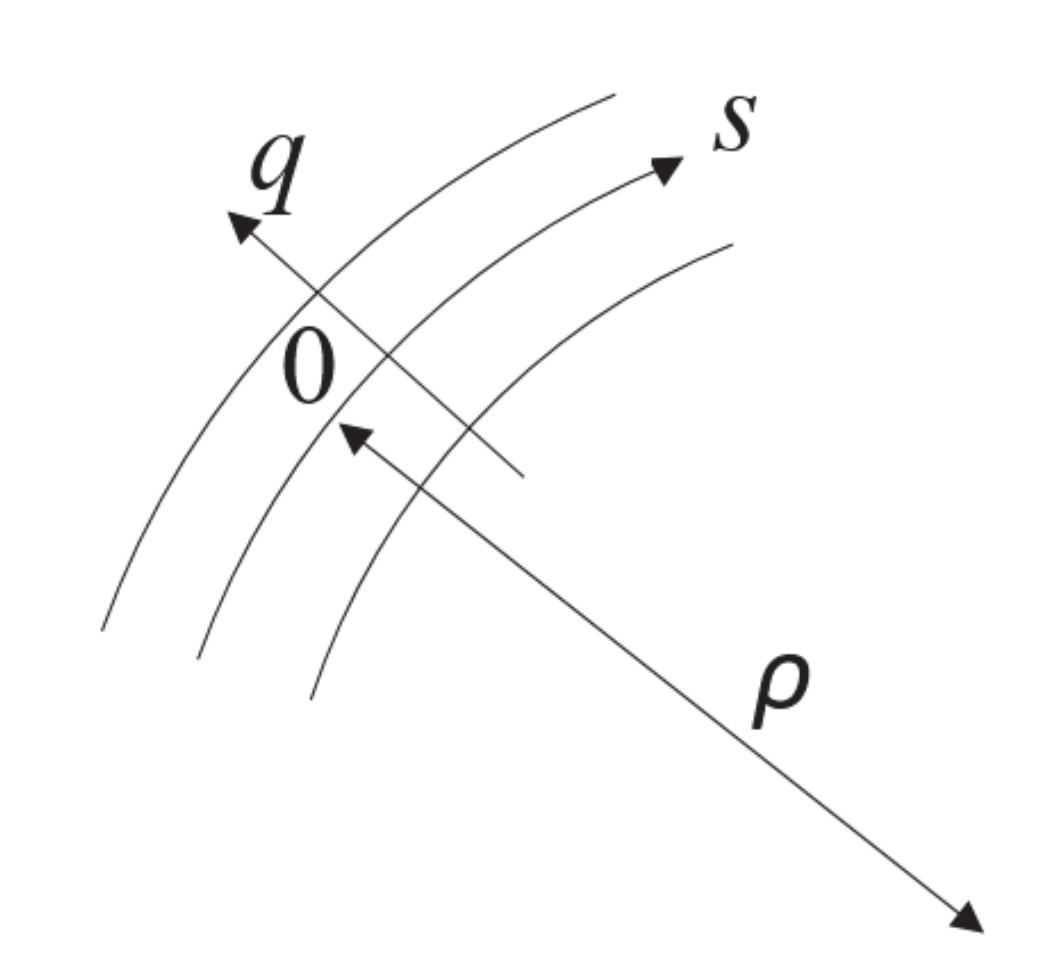}
\end{center}
\caption{Modeling of elastic beam}
\label{fig:model}
\end{figure} 

Thus in the experiment results mentioned in Section \ref{sec:Experiment},
we have considered three cases which have different thickness.

\section{Review of Euler's Elastica} \label{sec:elastica}

This section is devoted to 
 review the elastica theory following \cite{Mat13}.
\bigskip
\subsection{Geometry of a Curve in Plane}
\bigskip

\begin{figure}[ht]
\begin{center}
\includegraphics[width=0.45\hsize]{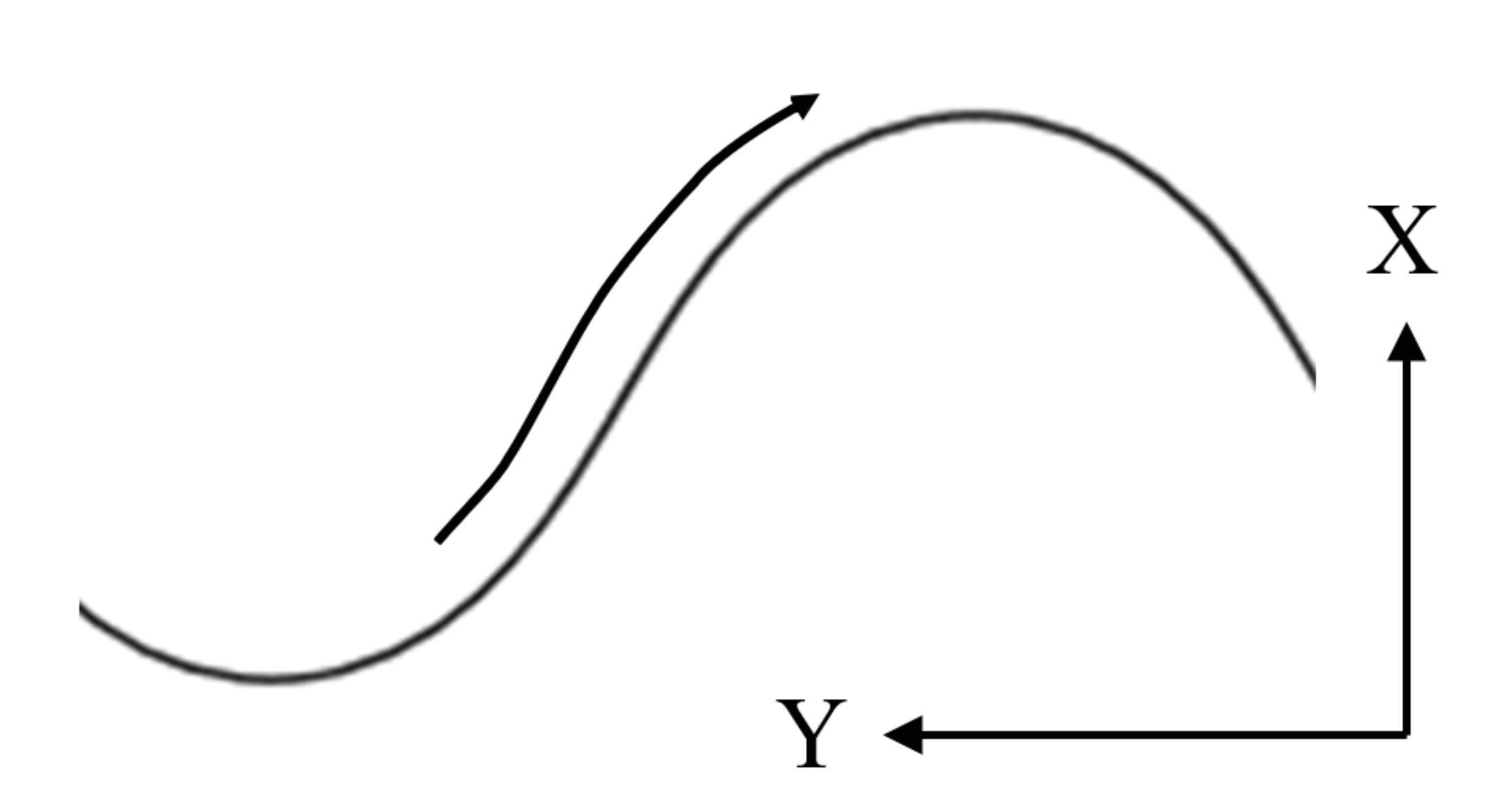}
\end{center}
\caption{Elastic curve}
\label{fig:elastica_Gamma}
\end{figure}

Let $Z: (s_1,s_2) \to \CC$ be an isometric analytic immersion
with the arclength $s$ for $s_1<s_2$.
In other words, we consider an analytic curve in a plane parameterized by the
arc-length $s$; $Z(s)=X(s) + \sqrt{-1}Y(s)$, i.e.,
$|\partial_s Z|=1$,
where $\partial_s:=\dd/\dd s$.
Its tangential vector
is  $\bt:=\partial_s Z=\ee^{\sqrt{-1}\varphi}$
using the tangential angle $\varphi \in
\{\varphi: (s_1, s_2) \to \RR \ |\ $ real analytic $\}$,
whereas the normal vector is
$\bn = \sqrt{-1}\bt$. 
We have the Frenet-Serret relation,
\begin{equation}
\partial_s(\partial_s Z) =\ii k \partial_s Z,
\label{eq:FSrelation}
\end{equation}
where $k:=\partial_s \varphi$ is the curvature
(inverse of curvature radius $\rho(s)$) of the curve.

\bigskip
From (\ref{eq:delta_k^2}),
the Euler-Bernoulli energy functional of $Z$
is given by
\begin{equation}
\cE[Z]=\frac{1}{2}\int_{(s_1,s_2)} k(s)^2\, \dd s.
\label{eq:energy} \end{equation} 
Let us consider its minimal point in the regular function space of $Z$,
$$
\fM_{(s_1,s_2)}:=\{ Z: (s_1,s_2) \to \CC \ |\ 
 Z\mbox{ is an isometric  analytic immersion}\},
$$
which is called celebrated {\it{elastica}} \cite{Griffiths, method, Mat13, Tr}.
In order to obtain the minimal point of the energy functional,
 we consider an infinitesimal deformation
$$
   Z_\varepsilon (s_\varepsilon) = Z(s) + \bn \varepsilon(s),
$$
which does not satisfy the isometric condition because
$$
   \partial_s Z_\varepsilon = (1 - \varepsilon k(s) ) \bt
          + (\partial_s \varepsilon) \varepsilon,
$$
and 
$$
   d s_\varepsilon^2 = d \overline{Z_\varepsilon} d Z_\varepsilon
      = (1 - 2 \varepsilon k) d s^2 + o(\varepsilon^2).
$$
The deformed curvature is given by
$$
k_\varepsilon = k + (k^2+\partial_s^2) \varepsilon
      + o(\varepsilon^2),
$$
since
$$
 \frac{\partial^2}{\partial s_\varepsilon^2}
Z_\varepsilon = (-(\partial_s \varepsilon)k) \bt
+(k +(k^2+\partial_s^2) \varepsilon) \bn
      + o(\varepsilon^2).
$$
The deformed integrated of the Euler-Bernoulli functional is given by
$$
k_\varepsilon^2 d s_\varepsilon
= (k^2 + (k^3+2k\partial_s^2) \varepsilon 
      + o(\varepsilon^2)) d s,
$$
and thus we have the following proposition:
\begin{proposition}
The curvature $k_\mm$ of the minimizer $Z_\mm$ of the Euler-Bernoulli energy 
functional {\rm (\ref{eq:energy})}, i.e.,
$\displaystyle{
Z_\mm|\ \min_{Z\in \fM_{(s_1,s_2)}} \cE[Z] }$,  satisfies 
\begin{equation}
    a k_\mm +\frac{1}{2} k_\mm^3 + \partial_s^2 k_\mm=0,
\label{eq:SMKdV} 
\end{equation}
where $a$ is a constant real number for the Lagrange multiplier.
We call $Z_\mm$ {\rm{elastica}} or {\rm{elastic curve}}.
\end{proposition}

\begin{proof}
The energy functional {\rm (\ref{eq:energy})} is reduced to 
\begin{equation}
    -\frac{\delta {\cE+a \int_{(s_1,s_2)} d s_\varepsilon}}
{\delta \varepsilon(s)}
       = k^3 +2 \partial_s^2k + a k=0
\end{equation}
since we consider the isometric deformation. 
\end{proof}

We note that there are uncountably infinite
 elasticas, $Z_\mm$'s, depending on their ending conditions.
From here we will consider only an element of 
the set $\fZ_\mm$ of elasticas, which is simply denoted by $Z$ again
in this section. The curvature $k_\mm$ is also simply denoted by $k$.

We have the governing equation of elastica:
\begin{proposition}
For a real constant $b$, the elastica obeys the equation
\begin{equation}
    (\partial_s k)^2 + \frac{1}{4}k^4 +
 a k^2 + b=0.
 \label{eq:C_1'}
\end{equation}
\end{proposition}
\begin{proof} By multiplying (\ref{eq:SMKdV}) by $(\partial_s k)$ and
integrating it, (\ref{eq:SMKdV}) becomes (\ref{eq:C_1'}). Here $b$
is an integral constant. Due to the reality of $k$ and $s$, $b$ must be also
real. \end{proof}

\subsection{Elastica in terms of elliptic functions}
\label{subsec:EllipticFunc}

For later convenience, we introduce affine parameters,
\begin{eqnarray} \label{eq:xyg1}
x(s) &:=&\frac{\ii}{4\alpha} \partial_s k + \frac{1}{8}k^2 +\frac{1}{12}a,
  \nonumber\\
&&\\
y(s) &:=& \frac{1}{2\alpha}\partial_s x =
\frac{1}{2}\left[
 \ii \left(\frac{1}{8} k^3 + \frac{1}{4}a k +\frac{\ii}{4\alpha} 
k \partial_s k
\right)
\right]
. \nonumber
\end{eqnarray}
(\ref{eq:C_1'}) means that we have
 an elliptic curve $C_1$ given by the affine equation,
\begin{gather}
\begin{split}
\frac{\hat{y}^2}{4}=y^2 &=
\left(x + \frac{1}{6}a\right)
\left(x - \frac{1}{12}a-\frac{1}{4}\sqrt{b}\right)
\left(x - \frac{1}{12}a+\frac{1}{4}\sqrt{b}\right)\\
&=(x-e_1)(x-e_2)(x-e_3),\\
\end{split}
\label{eq:C_1}
\end{gather}
where
$e_1 = -\frac{1}{6}a$,
$e_2 =  \frac{1}{12}a+\frac{1}{4}\sqrt{b}$, and
$e_3 =  \frac{1}{12}a-\frac{1}{4}\sqrt{b}$.
For later convenience, we let $a^2-b=16$;
$C_1=\{(x, y) \in \CC^2\ | \ $ (\ref{eq:C_1}) $\} \cup \{\infty\}$.
They mean that $a=2(e_2+e_3-2e_1)$, $b=-(e_2-e_3)^2$.
$(x,\hat y)$ corresponds to the Weierstrass normal form \cite{WW}.

For the curve $C_1$, the incomplete elliptic integral
of the first kind is given by
\begin{equation}
    u = \int^x_\infty \dd u, \qquad
    \dd u =  \frac{\dd x}{2y}.
\label{eq:int_1}
\end{equation}
The complete elliptic integrals of the first kind as
the double periodicity  $(2\omega', 2\omega'')$ are given by
$$
      \omega' := \int_{\infty}^{(e_1, 0)} \dd u, \qquad
      \omega'' := \int_{\infty}^{(e_3, 0)} \dd u,
$$
whereas the
complete elliptic integrals of the second kind are given by
\begin{equation}
      \eta' = \int_{\infty}^{(e_1, 0)} \dd r, \qquad
      \eta'' = \int_{\infty}^{(e_3, 0)} \dd r,
\label{eq:etas}
\end{equation}
where 
$$
    \dd r =  \frac{ x \dd x}{2y}.
$$
Using them, we define
the Weierstrass sigma function  $\sigma$ by
\begin{equation}
\sigma(u) = \frac{2\omega'}{2\pi\ii}
\exp\left(\frac{\eta' u^2}{2\omega'}\right)
\frac{\theta_1(u/\omega')}{\theta_1'(0)},
\label{eq:Wsigma}
\end{equation}
where $\tau=\omega''/\omega'$ and
$$
\theta_1(v) = \ii\sum_{n=-\infty}^{\infty}
        \exp\left(\ii \pi
              \left( \tau(n - 1/2)^2 + (2 n -1) (v + 1)
            \right)
            \right).
$$
In terms of the sigma function, the 
Weierstrass $\zeta$-function and $\wp$ function are given by
\begin{equation}
\zeta(u) = \frac{d}{d u} \log \sigma(u), \qquad
\wp(u) = -\frac{d^2}{d u^2} \log \sigma(u).
\label{eq:Wwp,zeta}
\end{equation}
We have an  identity between $\zeta$-function
and an integral of the second kind,
$$
    \zeta(u) = -\int^{(x,y)}_\infty \dd r
    = -\int^{(x,y)}_\infty x \dd u.
$$
Then it is known that 
$(\wp(u), \partial_u \wp(u)/2)$ is identified with
$(x,y)$ in $C_1$ by setting  $u = \int^{(x,y)}_\infty \dd u$;
we identify both by 
$x(s)= \wp(\alpha s+u_0)$ for a certain $u_0\in \CC$.

\subsection{Euler's Elastica and $\zeta$ function}
\label{subsec:EulerElasII}

Following \cite{Mat13}, we show the shape of elastica
as a minimizer of $\cE[Z]$ of $\fM_{(s_1,s_2)}$.
Here we do not consider the boundary condition explicitly
since $(s_1,s_2)$ has no boundary.

\begin{theorem} \label{thm:Z}
By choosing the origin of angle $\varphi$ and $u_0$,
\begin{equation}
    \partial_s Z(s) \equiv \ee^{\sqrt{-1}\varphi}=
       \ii(\wp(\alpha  s+u_0)-e_1),
\label{eq:pZ}
\end{equation}
\begin{equation}
    Z(s) = \frac{\ii}{\alpha}  (-\zeta(\alpha s + u_0)-e_1 s) + Z_0.
\label{eq:Z=zeta}
\end{equation}
\end{theorem}
\begin{proof}
Noting $\displaystyle{
    \ii k = \frac{\alpha \wp_u(\alpha s+u_0)}{\wp(\alpha s+u_0) - e_1}
}$ from  (\ref{eq:xyg1}),
the tangential angle of the elastica is given by
\begin{equation}
    \varphi(s) = \frac{1}{\sqrt{-1}}
                \log \Bigr(\wp(\alpha s+u_0) - e_1\Bigr) + \varphi_0.
 \label{eq:phi}
\end{equation}
It means that the tangential vector of elastica is represented by
an elliptic function and 
we have an explicit formula of $Z$ using the elliptic $\zeta$ function.
In other words, it is found that
$k\equiv\partial_s \varphi$ of (\ref{eq:xyg1})
 satisfies (\ref{eq:SMKdV})
and (\ref{eq:C_1'}) and vice versa.
\end{proof}

\begin{remark}
{\rm{
We have the following relation:
\begin{equation}
    X(u) =X_0 +\frac{\alpha}{4} k(s)
\label{eq:2.14}
\end{equation}
for an appropriate origin $X_0 \in \RR$. 
}}
\end{remark}

In the computation of elastica,  the condition that $\varphi$ and $s$
are real is necessary. We call the condition {\it{reality
condition}} i.e., $|\partial_s Z|=1$ and $s$ is real.

Let us call the tangential period $\hat \omega$ of the (open) elastica
that satisfies
$$
\partial_s Z\left(s + \frac{\hat \omega}{\alpha}\right)
 =\partial_s Z(s).
$$
Further we define an index of (open) elastica by
$$
\mathrm{index}(\partial_s Z) =
\frac{1}{2\pi \ii}
\left(\log\partial_s Z\left(s + \frac{\hat \omega}{\alpha}\right)
 - \log\partial_s Z(s)\right).
$$

Here we give a formula of the Euler-Bernoulli energy function;
\begin{proposition}\label{prop:3.5}
\begin{gather*}
\begin{split}
\frac{1}{2}\int_{s_1}^{s_2} k(s)^2 \dd s
=&\mathfrak{Re}\left(\frac{4}{\alpha}
\left(\zeta(\alpha s_1+u_0)-\zeta(\alpha s_2+u_0)\right)
-\frac{1}{3}a (s_2-s_1)\right),
\end{split}
\end{gather*}
where $\mathfrak{Re}(z)$ means the real part of $z$.
\end{proposition}

\begin{proof}
\begin{gather*}
\begin{split}
\frac{1}{2}\int_{s_1}^{s_2} k^2 \dd s
&=
4\int_{s_1}^{s_2}
\frac{1}{8} \wp(\alpha s+u_0) \dd s
-4\int_{s_1}^{s_2}\frac{\sqrt{-1}}{4\alpha}(\partial_s k)
-\frac{1}{3}a (s_2-s_1)\\
=&\frac{4}{\alpha}
\left(\zeta(\alpha s_1+u_0)-\zeta(\alpha s_2+u_0)\right)\\
&-\frac{\sqrt{-1}}{\alpha}( k(s_2)-k(s_1))
-\frac{1}{3}a (s_2-s_1).\\
\end{split}
\end{gather*}
Since $k$ is real, we have the expression.
\end{proof}

The number $\tau:= \omega''/\omega'$ is a complex number called modulus,
which determines the elliptic curve uniquely
modulo trivial transformation, translation, dilatation
and so on, and also determine the shape of elastica.

Due to the reality condition of the elastica, the moduli $\Xi$
of elastica is given by \cite{Mum}
\begin{equation}
\Xi := {\sqrt{-1}} \mathbb{R}_{>0} \cup \left(\frac{1}{2} + {\sqrt{-1}} \mathbb{R}_{>0}\right)
\cup \{\infty\}
\quad\mathrm{modulo}\quad \mathrm{PSL}(2,\mathbb{Z})
\label{eq:M_elas}
\end{equation}
as a subspace of the moduli of elliptic curves, $\Xi \subset
\mathbb{H}/ \mathrm{PSL}(2,\mathbb{Z})$, where $\mathbb{H}$ is the
upper half plane, i.e., $\mathbb{H}:=\{z \in \mathbb{C} \,
;  \Im z >0\}$.

\begin{figure}
\begin{center}
\includegraphics[width=0.45\hsize]{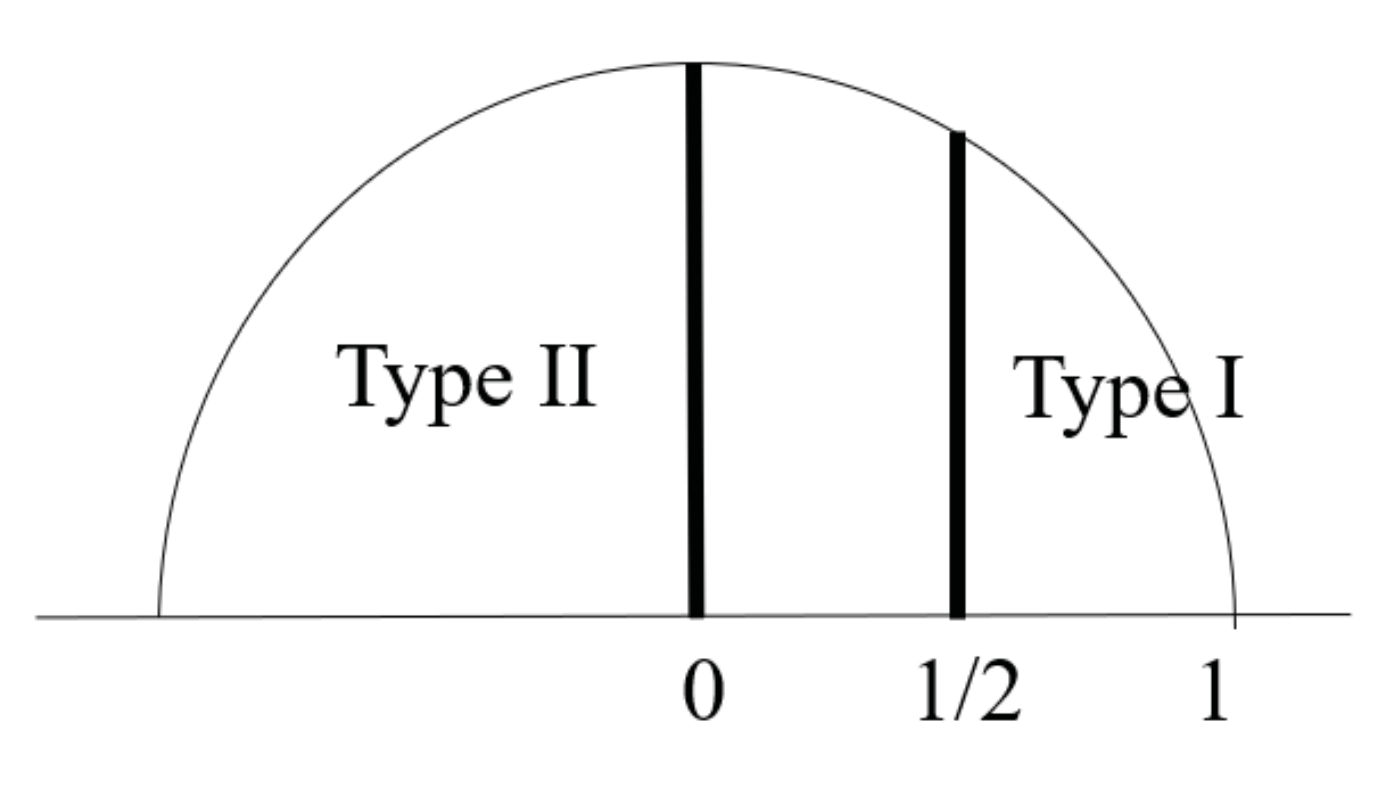}
\caption{Moduli $\Xi$}
\label{fig:typeI}
\end{center}
\end{figure}

This picture leads the classification of
elastica as follows, without proof \cite{Mat13,Mum};

\begin{figure}
\begin{minipage}{0.49\hsize}
\begin{center}
\includegraphics[angle=90, width=0.80\hsize]{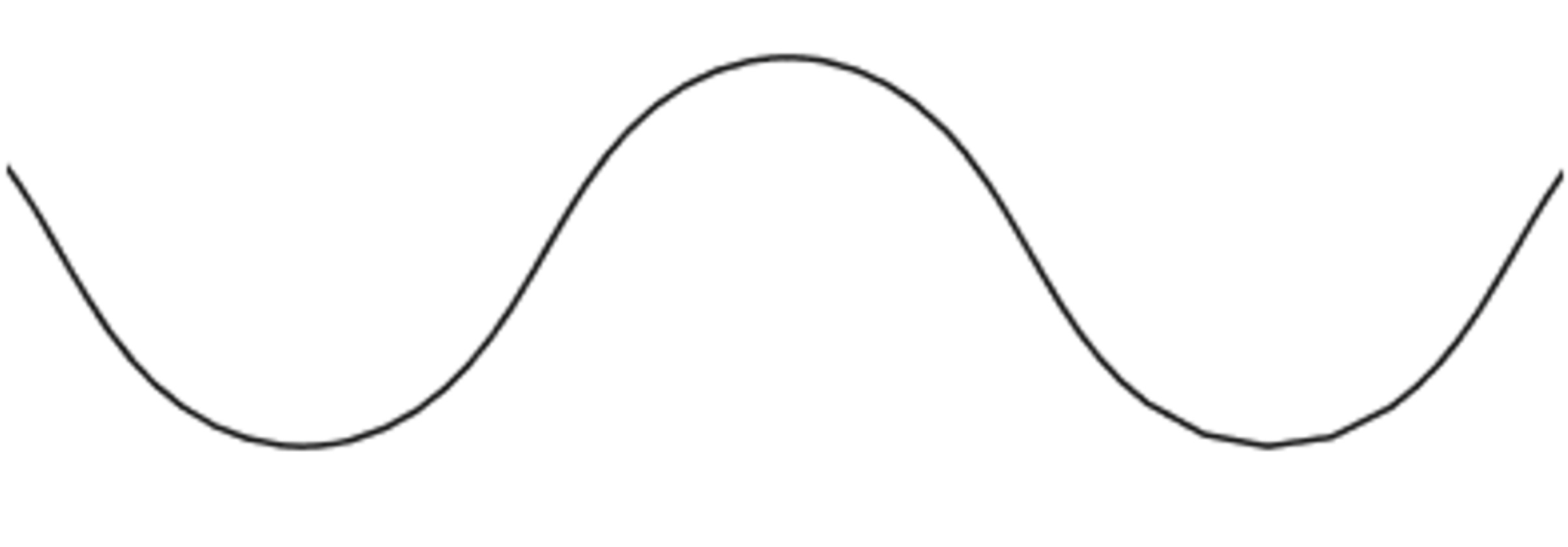}
\\
(a)
\end{center}
\end{minipage}
\begin{minipage}{0.49\hsize}
\begin{center}
\includegraphics[angle=90,width=0.80\hsize]{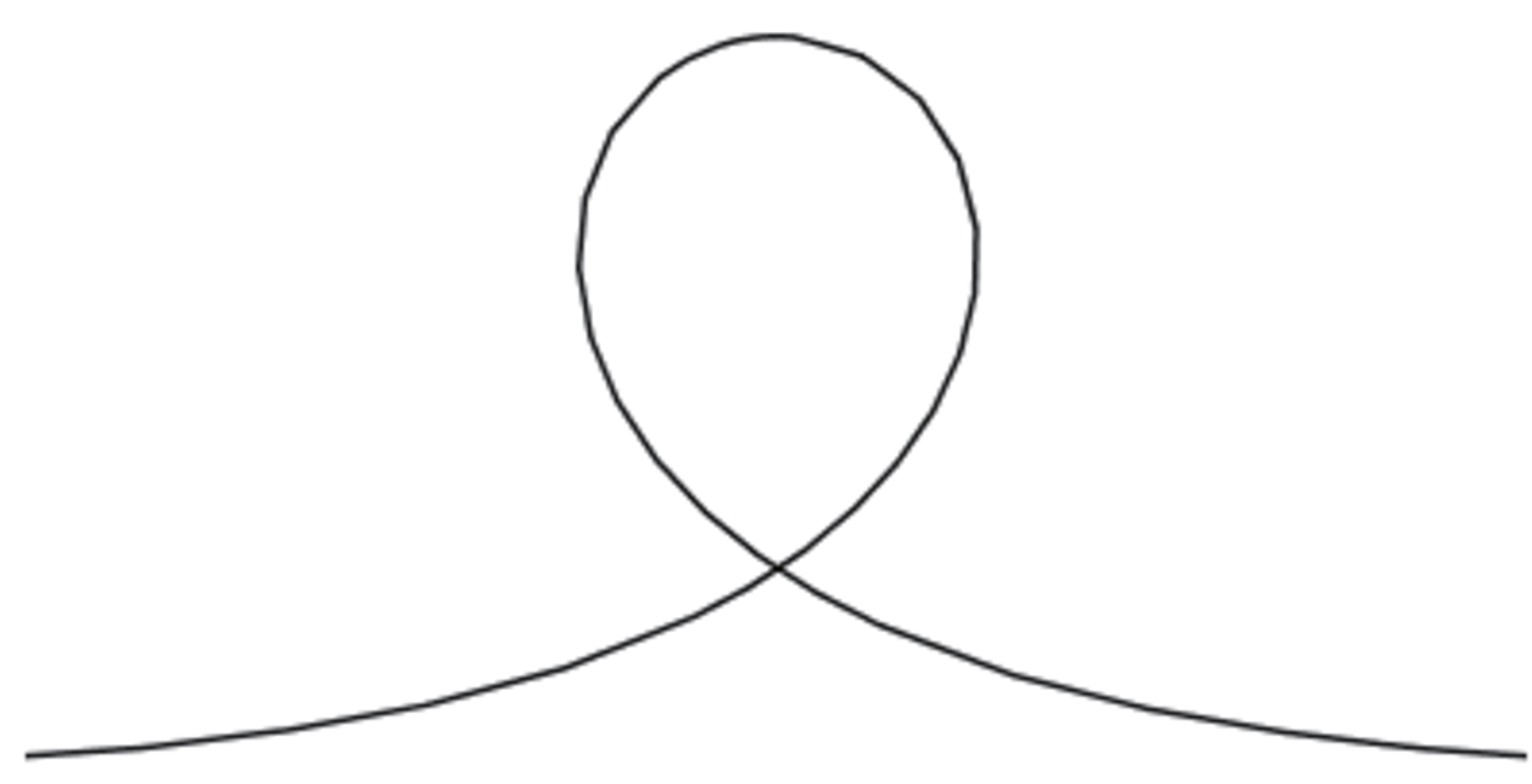}
\\
(d)
\end{center}
\end{minipage}\\
\begin{minipage}{0.49\hsize}
\begin{center}
\includegraphics[angle=90, width=0.80\hsize]{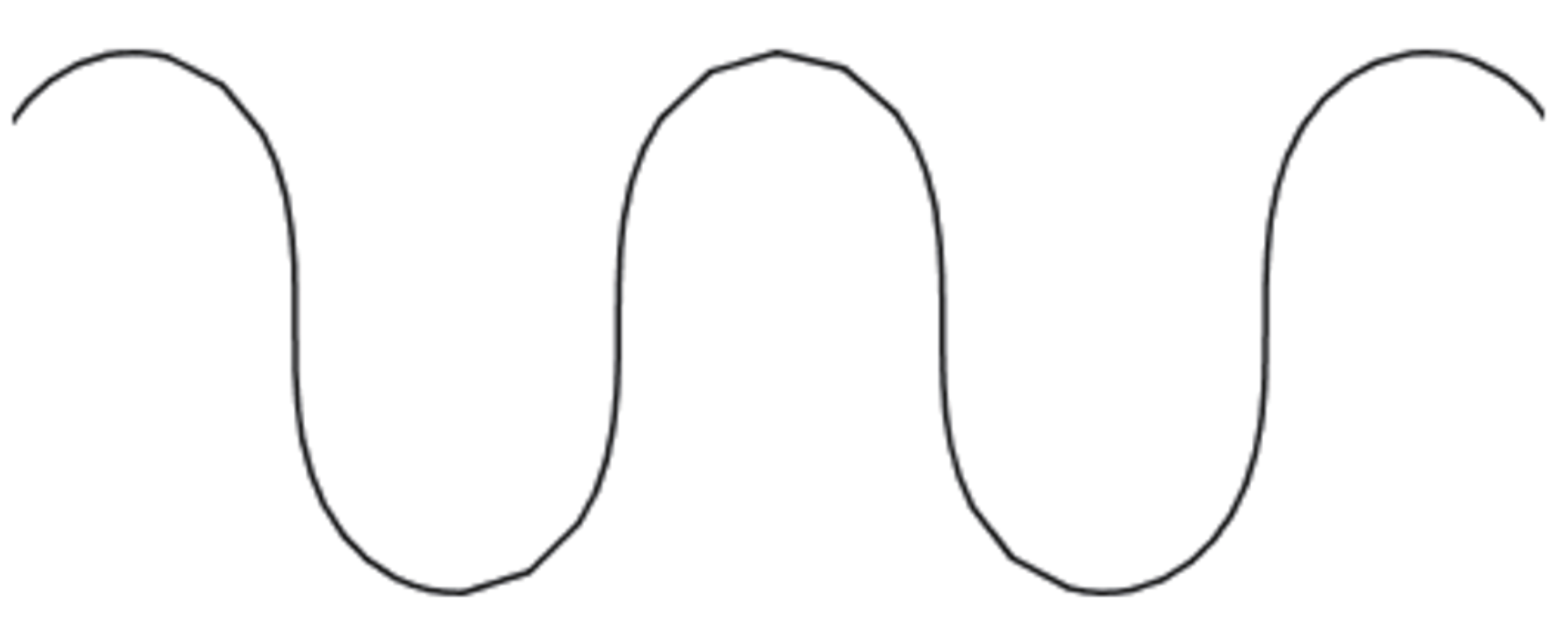}
\\
(b)
\end{center}
\end{minipage}
\begin{minipage}{0.49\hsize}
\begin{center}
\includegraphics[angle=90,width=0.80\hsize]{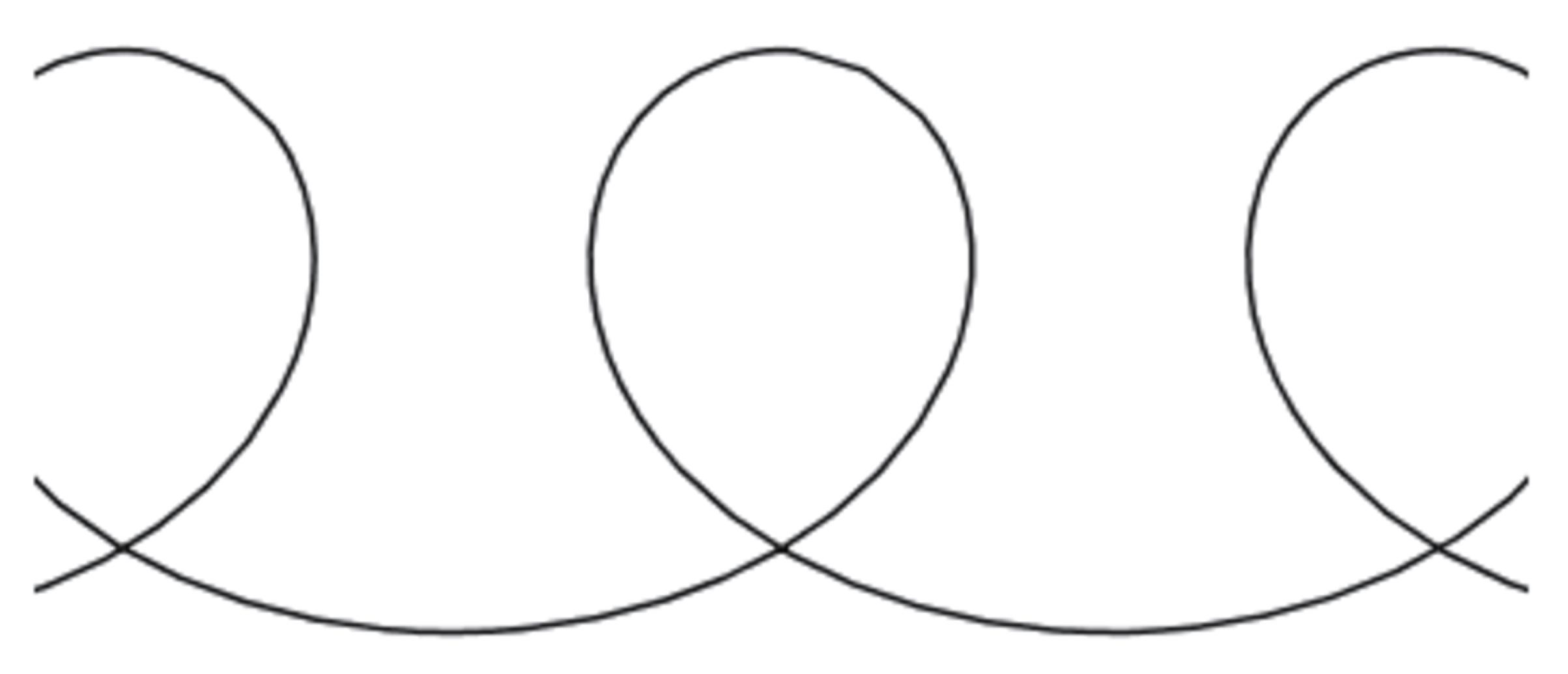}
\\
(e)
\end{center}
\end{minipage}\\
\begin{minipage}{0.499\hsize}
\begin{center}
\includegraphics[angle=90, width=0.99\hsize]{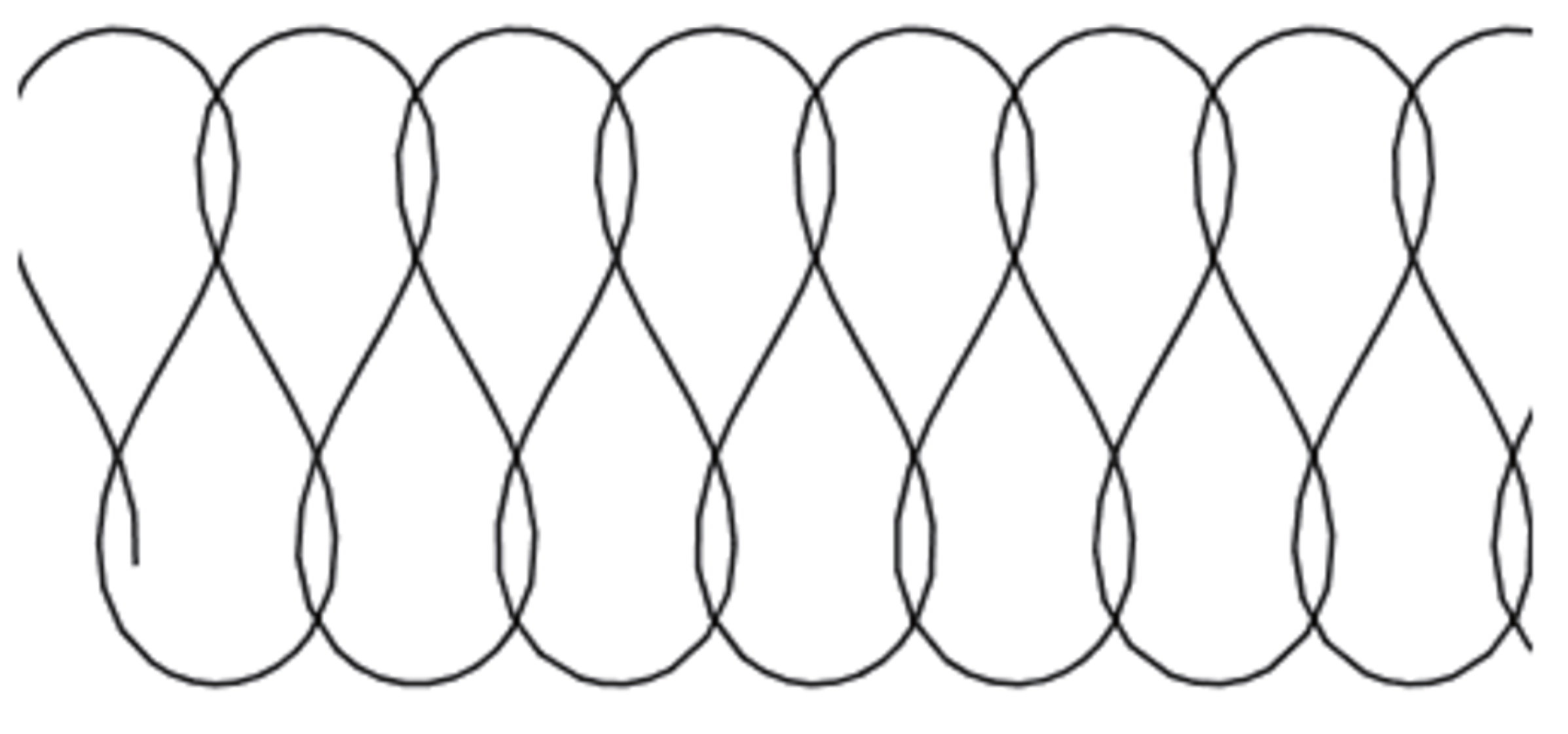}
\\
(c)
\end{center}
\end{minipage}
\begin{minipage}{0.49\hsize}
\begin{center}
\includegraphics[angle=90,width=0.99\hsize]{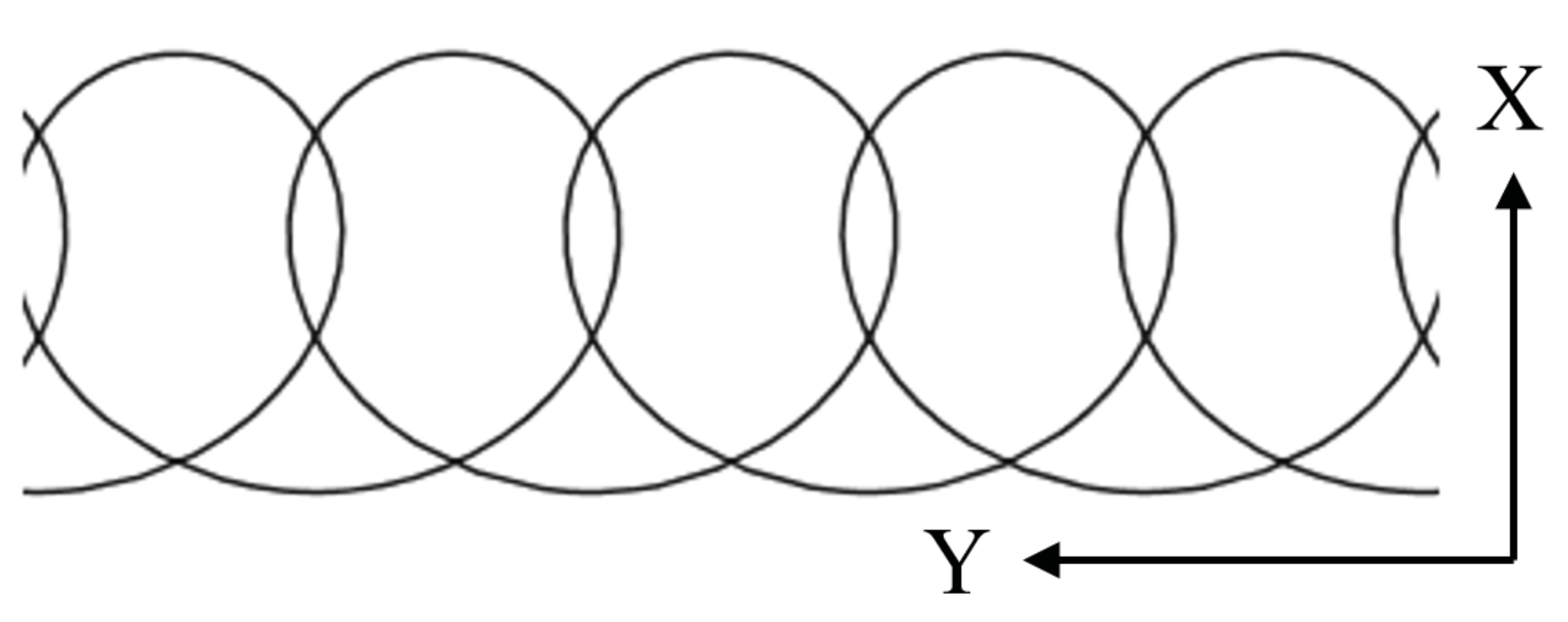}
\\
(f)
\end{center}
\end{minipage}\\
\caption{Types of Elastica:  (a) shows type Ia, (b) is the rectangular
elastica ($a=0$), (c) is type Ib, and (d)-(f) correspond to type II.}
\label{fig:typeI}
\end{figure}

\begin{proposition}{\rm{\cite{Mat13,Mum}}}
\label{prop:mod-elastica}
\begin{enumerate}

\item Type Ia: for the case $-4 \le a \le 0$, 
$\displaystyle{u_0 =
\left(\omega''-\frac{\omega'}{2}\right)}$,
$\hat \omega = 2 \omega'  \in \RR$,
$\displaystyle{\tau
\in\left( {\sqrt{-1}} \mathbb{R}_{>0} + \frac{1}{2}\right)}$
and $\mathrm{index}(\partial_s Z)$ is zero.

We call $a=0$ case, the rectangular elastica,
which corresponds to 
$\tau = \frac{1}{2}+\frac{1}{2}\ii$ and $1-\tau^{-1}=\ii$.

\item Type Ib: for the case $0 < a \le 4$, 
$u_0 = \displaystyle{ -\frac{\omega'}{2}}$,
$\hat \omega = 2 \omega' -4 \omega'' \in \RR$,
$\displaystyle{\tau
\in\left( {\sqrt{-1}} \mathbb{R}_{>0} + \frac{1}{2}\right)}$
and $\mathrm{index}(\partial_s Z) = 0$.

\item Type II: for $4 < a$, 
$u_0 = \displaystyle{\frac{\omega'}{2}}$,
$\hat \omega = 2 \omega'' \in \RR$,
 $\tau \in {\sqrt{-1}} \mathbb{R}_{>0}$
and $\mathrm{index}(\partial_s Z)$ is equal to
$\pm 1$.
\end{enumerate}
Here $\mathbb{R}_{>0}$ is $\{ x \in \mathbb{R} \, ;  x
>0\}$.
\end{proposition}

\bigskip

\section{Transition from elastica to $\Lambda_{\phi_0}$-elastica with hinge $\phi_0$.}\label{sec:Trans}

In this section,
we express the transition phenomenon from elastica to $\Lambda$-elastica.
In order to express it,
\begin{enumerate}

\item we explicitly express the boundary condition in the theory of elastica 
in Section \ref{sec:elastica} (We introduce the function space 
$\fM_{[s_1,s_2]}$ rather than $\fM_{(s_1,s_2)}$ and the 
boundary condition $\fB^{BT}_{W}$ with a parameter $W>0$.),

\item we introduce the novel function space $\fM^{s_0,\phi_0}_{[s_1,s_2]}$
in which the minimizer of the Euler-Bernoulli energy is $\Lambda$-elastica
of hinge $\phi_0$,

\item we prepare the function space $\fM^{s_0}_{[s_1,s_2]}$ which
includes the ordinary elasticas, $\fM_{(s_1,s_2)}$, and $\Lambda$-elasticas
$\fM^{s_0,\phi_0}_{[s_1,s_2]}$, and consider a disjoint orbit in 
$\fM^{s_0}_{[s_1,s_2]}$ as the transition, and 

\item using the symmetry, we set $s_1=-\frac{L}{2}$, $s_2=\frac{L}{2}$, $s_0=0$ and
give the explicit results of the transition.
\end{enumerate}

\subsection{Preliminaries}
From here, we discriminate the minimizer $Z_\mm$ and the general immersion $Z$.

Let $\rho^U_V$ be
 the restriction of the domain of the function from $U$ to $V(\subset U)$.
In order to impose the boundary condition, we consider
\begin{gather*}
\begin{split}
\fM_{[s_1,s_2]}:=&\Bigr\{Z: [s_1, s_2] \to \CC \ |
 \ Z \mbox{ is differentiable at } s_a (a=1,2), \\
&\rho^{[s_1,s_2]}_{(s_1,s_2)} Z \in \fM_{(s_1,s_2)}\Bigr\}.
\end{split}
\end{gather*}
For a real parameter $\phi_0$ and $s_0 \in (s_1, s_2)$,
 we introduce the function spaces,
\begin{gather*}
\begin{split}
\fM^{s_0,\phi_0}_{(s_1,s_2)}:=& \Bigr\{ Z: (s_1,s_2) \to \CC \ \Bigr| \ 
\mbox{continues, }\phi_0 = \displaystyle{\frac{1}{\ii}  
\log \frac{\partial_s Z(s_0+0)}{\partial_s Z(s_0-0)}},\\
& \rho^{(s_1,s_2)}_{(s_1,s_0)} Z \in \fM_{(s_1,s_0)},  
\rho^{(s_1,s_2)}_{(s_0,s_2)} Z \in \fM_{(s_0,s_2)} \Bigr\},
\end{split}
\end{gather*}
\begin{gather*}
\begin{split}
\fM^{s_0,\phi_0}_{[s_1,s_2]}:=&\Bigr\{Z: [s_1, s_2] \to \CC \ | 
 \ Z \mbox{ is differentiable at } s_a (a=1,2), \\
&\rho^{[s_1,s_2]}_{(s_1,s_2)} Z \in \fM^{s_0,\phi_0}_{(s_1,s_2)}\Bigr\},
\end{split}
\end{gather*}
and
$$
\fM^{s_0}_{[s_1,s_2]}:=
\bigcup_{\phi_0\in [0, 2\pi)} \fM^{s_0,\phi_0}_{[s_1,s_2]}.
$$

We have a simple relation.
\begin{lemma}
For given $s_0 \in (s_1, s_2)$,
\begin{equation}
\fM^{s_0,\phi_0}_{[s_1,s_2]}\subset \fM^{s_0}_{[s_1,s_2]} , \quad
\fM_{[s_1,s_2]}\subset \fM^{s_0}_{[s_1,s_2]}.
\label{eq:fMinfM}
\end{equation}
\end{lemma}

\subsection{Elastica with boundary condition}

In this subsection, we express the panel bending test by considering 
the boundary condition explicitly.  
For simplicity, we let $(s_1, s_2) = (-\frac{L}{2}, \frac{L}{2})$ and
introduce
the boundary condition $\fB_{BT}$ which corresponds to the bending test
in Section \ref{sec:Exp}, 
\begin{gather*}
\begin{split}
\fB^{BT}_{W}:=&\Bigr\{ Z: \left[-\frac{L}{2}, \frac{L}{2}\right] \to \CC\ 
\Bigr| 
 \ Z \mbox{ is differentiable at } \pm \frac{L}{2}, \ \\
 &Z\left(\pm \frac{L}{2}\right) = X_0 + \frac{W}{2} \ii, \ 
\partial_s Z\left(\pm \frac{L}{2}\right) = \ii\Bigr\},
\end{split}
\end{gather*}
where $W(>0)$ means the width of the ending of the elastica $Z$.
The shape $Z_\mm$ of the ordinary
 elastica in the compression testing apparatus
is obtained as the minimizer 
$$
  Z_\mm^W| \min_{Z \in \fM_{[-\frac{L}{2}, \frac{L}{2}]}
\bigcap \fB^{BT}_{W} } \cE[Z].
$$
We obviously have the simple result;
\begin{lemma}
$Z_\mm^L([-\frac{L}{2}, \frac{L}{2}])
=\{X_0 + s\ii\ |\ s\in [-\frac{L}{2}, \frac{L}{2}]\}$.
\end{lemma}

It is noted that 
for $W \in (0, L]$, there are two points $Z_\mm^W$, which are up-concave 
and down-concave.
We are concerned only with a continuous deformation from the 
straight elastica $Z_\mm^L$. We will choose the down-concave shapes.
We consider one parameter deformation in 
$\fM_{[-\frac{L}{2}, \frac{L}{2}]}$
for a deformation parameter $t\in I:=[0,1)$ with compression,
$$
 w(t) = (1-t)\cdot L.
$$
Let us consider a continuous orbit in $\fM_{[-\frac{L}{2}, \frac{L}{2}]}$,
$$
Z_\co^{w} : I \to \fM^{0}_{[-\frac{L}{2}, \frac{L}{2}]}, \quad
\left(t\mapsto Z_\co^{w}(t)=Z_\mm^{w(t)} 
\in \fM_{[-\frac{L}{2}, \frac{L}{2}]}\right).
$$
Since $Z_\co^{w}$ is continuous and $Z_\co^w(0)=Z_\mm^L$, 
$Z_\co^{w}$ is given by the following lemma.
\begin{lemma}\label{lm:4.3}
For $a \in [-4,0]$,
$Z_\co^w(t)(s) =  \frac{\ii}{\alpha}  (-\zeta(\alpha s + u_0)-e_1 s) + X_0$
where
$\displaystyle{u_0 =
\left(\omega''-\frac{\omega'}{2}\right)}$,
$\hat \omega = 2 \omega'  \in \RR$ and
$\displaystyle{\alpha=\frac{\hat \omega}{L}}$
such that $w(t)=(Z_\co^w(t)(L/2)-Z_\co^w(t)(-L/2))/\ii$.
\end{lemma}
The case $a=-4$ corresponds to $Z_\mm^L$ and $t=0$ whereas
the case $a=0$ corresponds to the part of the rectangular elastica
and $t=t_R:=0.8767723366$.
$Z_\co^{w}$ expresses the deformation in the panel bending test
and 
Lemma \ref{lm:4.3} shows the behavior of $Z_\co^{w}$ for $t\in [0, t_R]$.

For the elastica $Z_\co^w$, we denote its curvature by $k_\co^w$.
Since the curvature $|k_\co^w(t)(s=0)|$ is the largest curvature 
in the elastic curve $Z_\mm$,
we fix the point $s_0$ by $s_0 = 0$
using the symmetry for the boundary condition.

\bigskip

\subsection{$\Lambda_{\phi_0}$-elastica}

With a certain boundary condition, the minimizer $Z_\mm$ of the Euler-Bernoulli
functional 
$$
\cE^{\Lambda_{\phi_0}}[Z]:=
\cE[\rho^{(-\frac{L}{2}, \frac{L}{2})}_{(-\frac{L}{2}, 0)}Z]
+\cE[\rho^{(-\frac{L}{2}, \frac{L}{2})}_{( 0,\frac{L}{2})}Z]
$$
 in $\fM^{s_0,\phi_0}_{[s_1,s_2]}$ is the $\Lambda_{\phi_0}$-elastica.
We investigate it in this subsection.


Let us consider a disjoint orbit in $\fM^{0}_{[-\frac{L}{2}, \frac{L}{2}]}$
as a transition from elastica to $\Lambda$-elastica.
For a positive parameter $k_\cc$, which we call critical curvature,
we define the critical time $t_\cc^{k_\cc}$ by
$$
t_\cc^{k_\cc} := \sup_{t \in I}
\left\{|k_\co^w(t)(0)| < k_\cc 
     \right\}.
$$
We have the critical width
$$
W_\cc:=w(t_\cc^{k_\cc})=(1-t_\cc^{k_\cc})\cdot L.
$$
Then we can express the transition from  elastica
to $\Lambda_{\phi_0}$-elastica as
$$
Z^{k_\cc,\phi_0}_{\dor} : 
I_\cc \to \fM^{0}_{[-\frac{L}{2}, \frac{L}{2}]},
$$
where $I_\cc:=[t_\cc^{k_\cc}, 1]$ and the disjoint orbit,
$$
Z^{k_\cc,\phi_0}_{\dor}(t):=
\left\{ \begin{array}{ll}
Z_\co^w(t)
  &  \mbox{for} t < t_\cc^{k_\cc}\\
Z_\mm| \displaystyle{
\min_{Z \in \fM^{0,\phi_0}_{[s_1,s_2]}\bigcap \fB^{BT}_{w(t)} } 
\cE^{\Lambda_{\phi_0}}[Z]} , &  \mbox{for  } t = t_\cc^{k_\cc}
\end{array}\right..
$$

The following proposition is obvious:
\begin{proposition}\label{prop:L-elastica}
The minimizer $Z_\mm$ of 
$\cE^{\Lambda_{\phi_0}}[Z]$ in $\fM^{s_0,\phi_0}_{[s_1,s_2]}$
consists of the parts of elastica.
\end{proposition} 

\begin{remark}
{\rm{
The $\Lambda_{\phi_0}$-elastica can be regarded as a curve of 
picewised elastica.
Thus we can apply the Weierstrass-Erdmann corner conditions to
this system directly \cite{GF}, though we employ another approach.
}}
\end{remark}

Following Proposition \ref{prop:L-elastica}, 
we numerically compute $\Lambda_{\phi_0}$-elastica.
For $\phi_0=\pi/4$, the numerical computations shows 
a disjoint orbit $Z^{k_\cc,\phi_0}_0(t)$ 
illustrated in Figure \ref{fig:Belastica_Gamma}.
We set $L=1$.

In the computation, we used the Maple 2019.
We assume that 
$\Lambda^{\phi_0}$-elastica consisting of type II elastica
 in 
Proposition \ref{prop:mod-elastica} 
is the minimizer of $\cE^{\Lambda_{\phi_0}}[Z]$.
In other words,
we searched the minimal point only in type II elastica
for $\Lambda^{\phi_0}$-elastica,
even though there are other local minimal points in the function space
because
the shape which satisfy the boundary condition and is given by type II elastica
obviously 
seems to have smaller curvature than the shapes consisting of other type 
elastica;
we do not argue the other possibilities in this paper.

We fix the parameter $a\in [4, \infty)$ in type II elastica.
From Proposition \ref{prop:mod-elastica}, we find $\alpha s_1$ and 
$\alpha s_2$ so that these points correspond to the minimal $X_\mm$, e.g., in
Figure \ref{fig:typeI} (f), which satisfies the boundary condition 
$\partial_s Z_\mm(s_i) = \ii$, $(i=1,2)$ using (\ref{eq:pZ}).
We numerically found $\alpha s_0$ for the transcendental equation,
$$
\log (\partial_s Z_\mm(s_0)) = \phi_0\ii
$$
using (\ref{eq:phi}).
It determines $\alpha$ because of
$\alpha (s_1-s_0)=L/2$ and then we have the width $W_{II}$
as a function of $a$.
Thus for a given width $W_\cc$,
using the bisection method, we found $a$ which reproduces $W_\cc$
up to a certain error.

The shape of $\Lambda_{\phi_0}$-elastica and the transition is given in 
Figure \ref{fig:Belastica_Gamma} (c)-(a).
We define the energy gap  by
\begin{equation}
\Delta \cE^{\Lambda_{\phi_0}}
:= \lim_{t\to t_\cc-0} \cE[Z^{k_\cc,\phi_0}_0(t)]
-\cE^{\Lambda_{\phi_0}}[Z^{k_\cc,\phi_0}_0(t_\cc^{k_\cc})].
\label{eq:DeltaE}
\end{equation}
In the case of Figure \ref{fig:Belastica_Gamma}, 
$\Delta \cE^{\Lambda_{\phi_0}}$ is positive,
as in Figure \ref{fig:Belastica_Gamma} (f).
It means that by the transition, total energy 
decreases; the $\Lambda$-elastica is  stabler than  
the ordinary elastica.

\begin{figure}[ht]
\begin{minipage}{0.38\hsize}
\begin{center}
\includegraphics[width=0.9\hsize]{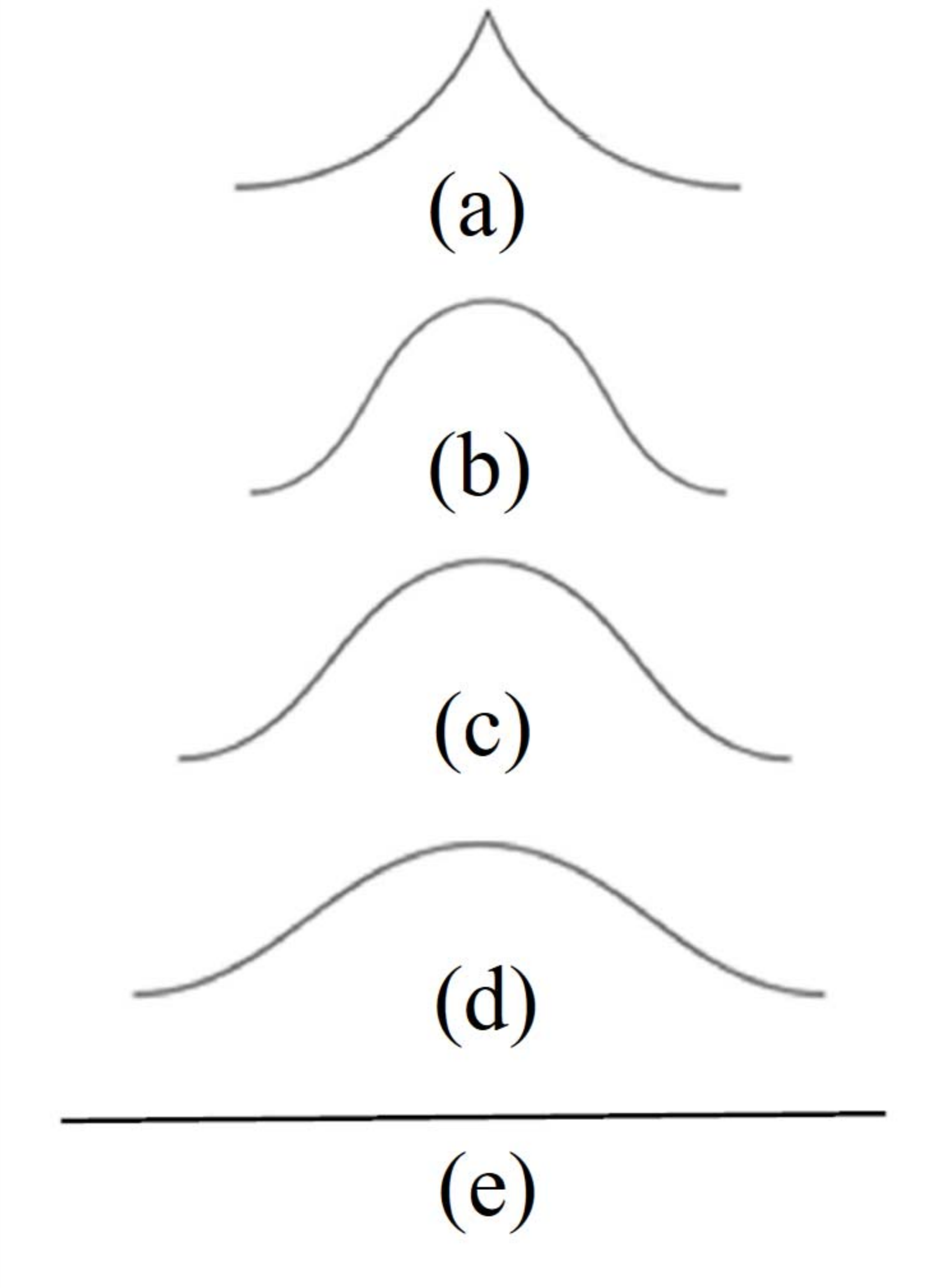}
\end{center}
\end{minipage}
\begin{minipage}{0.61\hsize}
\begin{center}
\includegraphics[width=0.95\hsize]{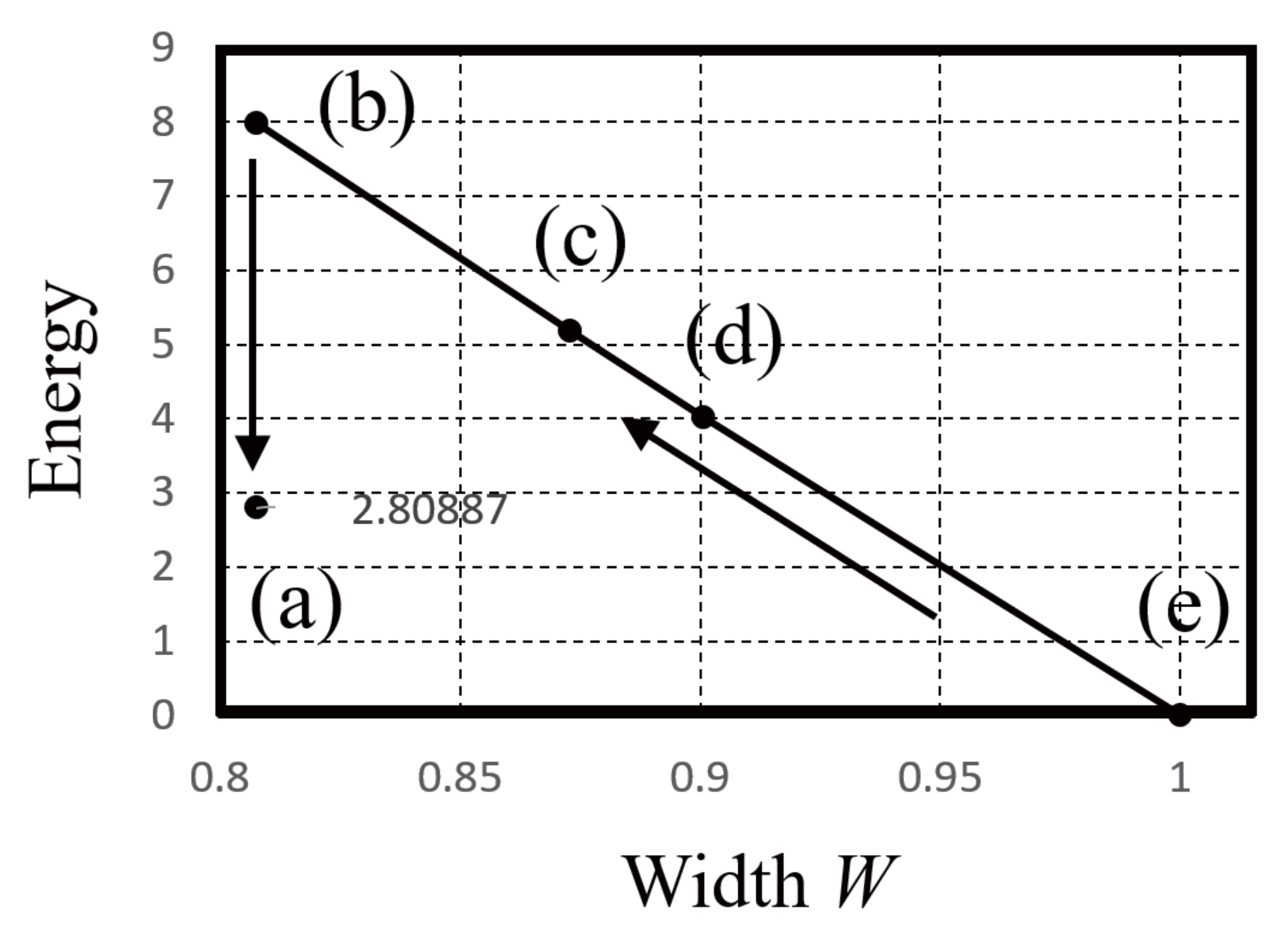}
\\
(f)
\end{center}
\end{minipage}
\caption{A transition: (a)-(e) shows 
the orbit from elastica to $\Lambda_{\phi_0}$ elastica;
from (e) to (a), whereas for the orbit, the total energy is illustrated in
(f). The width of (b) is the same as that of (a), which corresponds to
the critical width. In the computation, we let $L=1$ and $\kappa=1$.}
\label{fig:Belastica_Gamma}
\end{figure} 

Under this boundary condition, 
(\ref{eq:Z=zeta}) and Proposition \ref{prop:3.5} mean
that the relation between 
the width $W$ and the energy $\cE[Z_\mm^W]$ is given as a linear 
equation,
\begin{equation}
\cE[Z_\mm^W] = E_0 (L-W)
\label{eq:ELW}
\end{equation}
because both are written by the Weierstrass' zeta functions.

It is obvious to have the positivity of the energy gap
from the fact (\ref{eq:fMinfM}):
\begin{proposition}
$$
\Delta \cE^{\Lambda_{\phi_0}}:= \sup_{\phi_0 \in [0, \pi)} \Delta \cE^{\Lambda_{\phi_0}}
$$
is non-negative.
\end{proposition}
However for given $\phi_0$ and $W_\cc$,
the positiveness of 
$\Delta \cE^{\Lambda_{\phi_0}}$ is not guaranteed but
there exists $\phi_0$ whose $\Delta \cE^{\Lambda_{\phi_0}}$ is non-negative
since the case $\phi_0=\pi$ corresponds to ordinary elastica.
It might be expected that $\phi_0$ should be determined as
a minimal point of the energy $\Delta \cE^{\Lambda_{\phi_0}}$
in the parameter space $\phi_0\in [0, \pi]$.

We computed the cases with several conditions of $W_\cc$ and 
$\phi_0 = 0, \pi/4, \pi/2$ numerically and
draw up lists of them as in Figures \ref{fig:Belastica_final0}
and \ref{fig:Belastica_final1}, and Table 2-6.
Figure \ref{fig:Belastica_final1} shows the table of the elastica and 
$\Lambda_{\phi_0}$-elastica.
The blank in Figures \ref{fig:Belastica_final1} and Table 2-6
 means that we cannot find the 
$\Lambda_{\phi_0}$-elastica;
more precisely we can find shape which satisfies the boundary conditions
at $s_1$, $s_0$ and $s_2$ but since it has the much higher energy, we 
do not employ it as $\Lambda_{\phi_0}$-elastica in this paper.
In this computation, we also used the algorithm as mentioned above.
Table 2 shows the computed width of each shape by the bisection method.
Table 3 shows the height $X_\cc$ and $X_\Lambda$ for every width $W$.
Table 4 shows the elastica parameter $a$ of each shape
and Table 5 shows the imaginary part $\tau_i$ of the moduli parameter $\tau$
of elliptic function, i.e.,
 $\tau_i$ of $\tau =1/2 +\tau_i \ii$ for the elastica
and $\tau=\tau_i\ii$ for the $\Lambda_{\phi_0}$-elastica.
Table 6 gives each energy $\cE[Z]$ and 
$\cE^{\Lambda_{\phi_0}}[Z]$.
We display the results in Figure \ref{fig:Belastica_final0}.

\begin{table}[htb]
Table 2: Width $W_\cc$ computed by means of the bisection method
  \begin{tabular}{|c|r|r|r|r|r|}
\hline
$\phi_0$ & $W_5$ & $W_4$ &  $W_3$ & $W_2$ &$W_1$ \\
\hline
$\pi/2$&  &   &  & 0.872 & 0.900\\
\hline
$\pi/4$&  &   &0.807  & 0.873 & 0.901\\
\hline
$0$& 0.648 & 0.742  &0.808  & 0.872 & 0.901\\
\hline
elastica& 0.645  &  0.742& 0,808& 0.873 & 0.901\\
\hline
  \end{tabular}
\end{table}

\begin{table}[htb]
Table 3: Height $X_\cc$ and $X_\Lambda$ 
  \begin{tabular}{|c|r|r|r|r|r|}
\hline
$\phi_0$ & $W_5$ & $W_4$ &  $W_3$ & $W_2$ &$W_1$ \\
\hline
$\pi/2$&  &   &  & 0.186 & 0.186\\
\hline
$\pi/4$&  &   &0.245  & 0.186 & 0.154\\
\hline
$0$& 0.313 & 0.261  &0.212  & 0.151 & 0.119\\
\hline
elastica& 0.334  &  0.297& 0,262& 0.218 & 0.194\\
\hline
  \end{tabular}
\end{table}

\begin{table}[htb]
Table 4:Elastica parameter $a$ in the computations
  \begin{tabular}{|c|r|r|r|r|r|}
\hline
$\phi_0$ & $W_5$ & $W_4$ &  $W_3$ & $W_2$ &$W_1$ \\
\hline
$\pi/2$&  &   &  & 45.0 & 10000.0\\
\hline
$\pi/4$&  &   &6  & 4.08 & 4.0148\\
\hline
$0$& 20.0 & 4.45  &4.06  & 4.0023 & 4.000167\\
\hline
elastica& -1.3  &  -2& -2.5& -3 & -3.215\\
\hline
  \end{tabular}
\end{table}

\begin{table}[htb]
Table 5:Imaginary part $\tau_i$ of the moduli parameter $\tau$
  \begin{tabular}{|c|r|r|r|r|r|}
\hline
$\phi_0$ & $W_5$ & $W_4$ &  $W_3$ & $W_2$ &$W_1$ \\
\hline
$\pi/2$&  &   &  &  0.3492 & 0.1586 \\
\hline
$\pi/4$&  &   & 0.6080 & 1.1749 & 1.4429 \\
\hline
$0$& 0.4272 & 0.9036 & 1.2205 & 1.7390 &2.1560 \\
\hline
elastica& 0.5826 & 0.6396 & 0.6914 & 0.7617 & 0.8902 \\
\hline
  \end{tabular}
\end{table}

\begin{table}[htb]
Table 6: Energy $\cE[Z]$ and 
$\cE^{\Lambda_{\phi_0}}[Z]$
  \begin{tabular}{|c|r|r|r|r|r|}
\hline
$\phi_0$ & $W_5$ & $W_4$ &  $W_3$ & $W_2$ &$W_1$ \\
\hline
$\pi/2$&  &   &  & 1.23 & 1.23\\
\hline
$\pi/4$&  &   &2.81  & 3.67 & 4.60\\
\hline
$0$& 4.94 & 5.64  &7.22  & 10.76 & 13.81\\
\hline
elastica& 15.45 &  10.96& 7.99& 5.18 & 4.03\\
\hline
  \end{tabular}
\end{table}

\begin{remark}
{\rm{
For given $\phi_0$ and $W_\cc$,
the positiveness of 
$\Delta \cE^{\Lambda_{\phi_0}}$ is not guaranteed  
as in Figure \ref{fig:Belastica_final0}.
Figure \ref{fig:Belastica_final0} shows that 
in many cases, $\Delta \cE^{\Lambda_{\phi_0}}$ is positive
whereas there exist the case
 in which $\Delta \cE^{\Lambda_{\phi_0}}$ is negative.

We assume that for given $k_\cc$ and $\phi_0$, 
 $Z^{k_\cc,\phi_0}_0(t_\cc^{k_\cc})$ consists of elastica of 
type II,
though we did not compare the other local minimum of the elastica which 
has the boundary condition. Then 
the transition from  elastica to $\Lambda$-elastica
is given by 
a map in the moduli space of the elastica as in Table 5.
 It is quite interesting from
the viewpoint of the study on the moduli of elastica.
}}
\end{remark}

\begin{figure}
\begin{center}
\includegraphics[width=0.7\hsize]{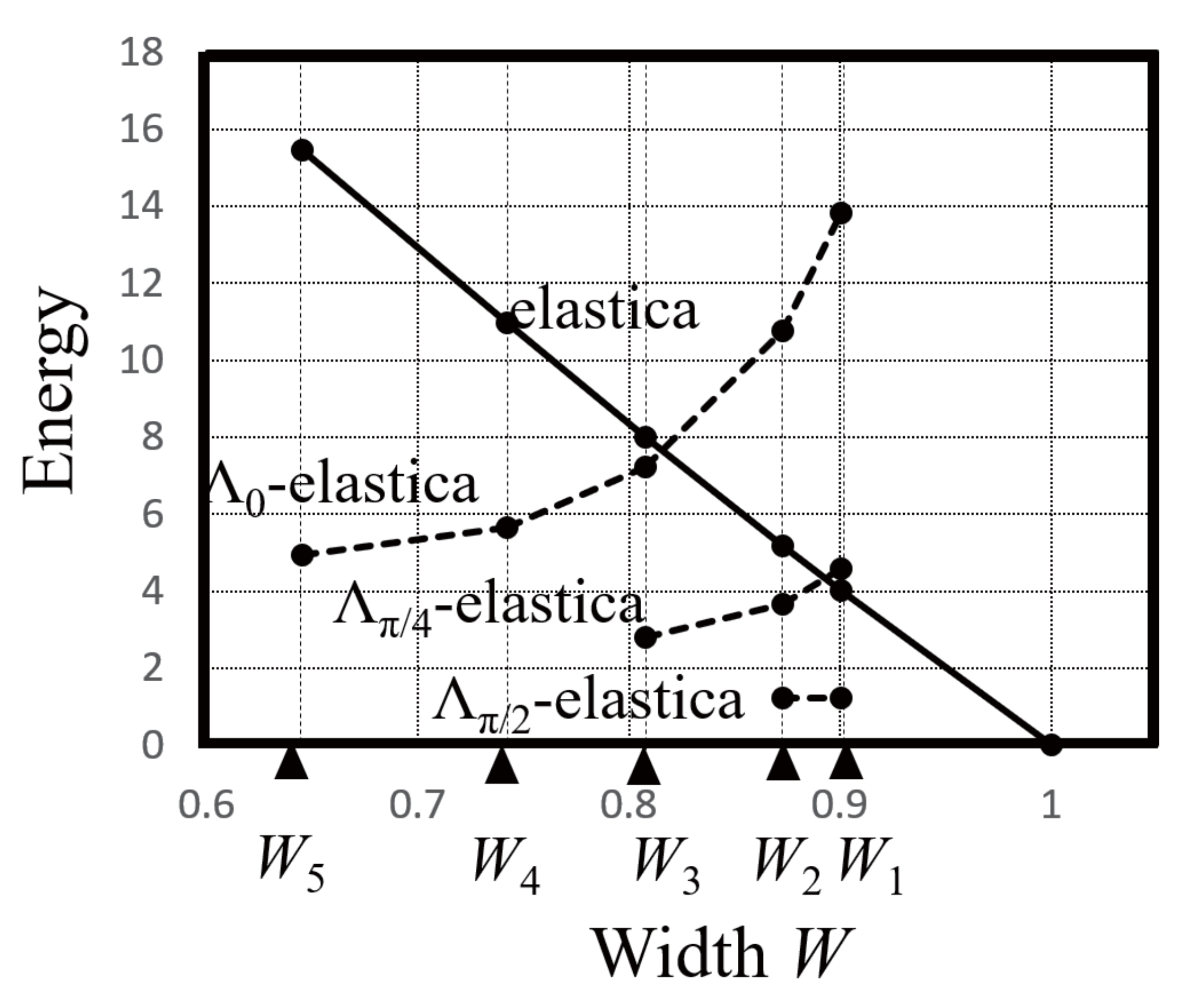}
\caption{$\Lambda_{\phi_0}$-elastica and energy in Table 2-6}
\label{fig:Belastica_final0}
\end{center}
\end{figure}

\begin{figure}
\begin{center}
\includegraphics[width=\hsize]{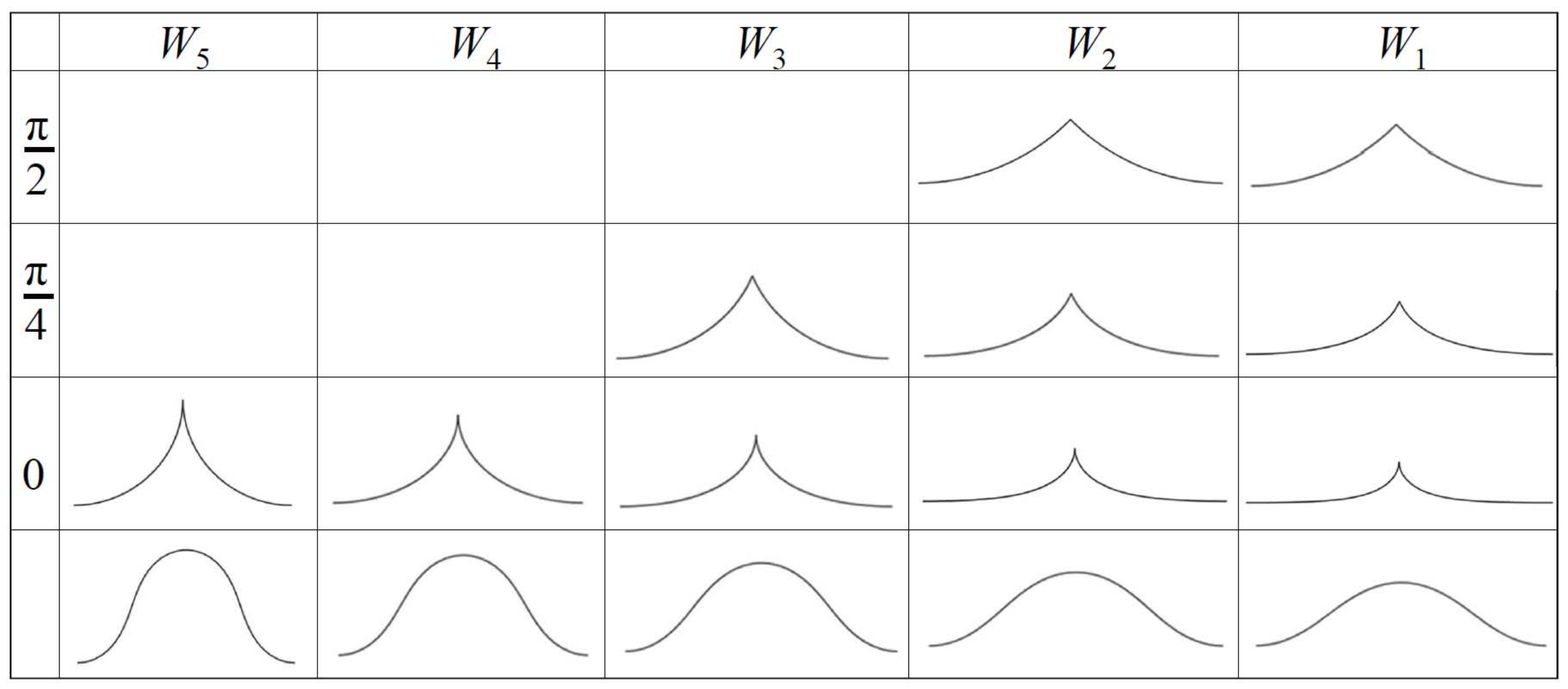}
\caption{Elastica and 
$\Lambda_{\phi_0}$-elastica in Table 2-6}
\label{fig:Belastica_final1}
\end{center}
\end{figure}

\section{Discussion}\label{sec:Dis}

In this paper, we investigated the 
$\Lambda$-elastica. 
We explicitly show the shape of $\Lambda_{\phi_0}$-elastica
in terms of Weierstrass elliptic $\zeta$-functions,
and numerically 
showed it in Figures
\ref{fig:Belastica_Gamma} and \ref{fig:Belastica_final1}.
By estimating their energy, we also considered
the 
transition from  elastica to $\Lambda$-elastica and stability 
from the viewpoint of energetic study.
The energy gap $\Delta \cE^{\Lambda_{\phi_0}}$ in (\ref{eq:DeltaE}) 
are numerically computed and illustrated in 
Figure \ref{fig:Belastica_final0} and Table 6.

When we compare the computational results and 
experimental results in Section \ref{sec:Exp},
the effective elastic constant is crucial, which is proportional to 
the thickness $\delta$, whereas
we used
the normalized elastic constant, $\kappa = 1$, in the computations.
The thickness $\delta$ of the elastic panel has the energy
$\delta \cdot \Delta \cE^{\Lambda_{\phi_0}}$ and,
for examples, the values in the graph of
Figure \ref{fig:Belastica_final0} should be multiplied by its 
thickness $\delta$.
In order to consider the effect of $\delta$, 
Table 1 reads the following table, Table 7.

\begin{table}[htb]
Table 7: The thickness vs $X_\cc$, $W_\cc$, $\phi_0$
and $X_\Lambda$
 in  Figure \ref{fig:experiment} 
  \begin{tabular}{|r|r|r|r|r|r|r|r|}
\hline
$\delta$ & $X_\cc$ & $\delta\cdot X_\cc$ &  $W_\cc$ & $\delta(L-W_\cc)$
&$\phi_0$ & $X_\Lambda$ & $\delta\cdot X_\Lambda$ \\
\hline
2.0 [mm]& 49[mm]   &98[mm${}^2$] & 234[mm]& 51[mm${}^2$]& 0.66$\pi$&
51[mm]   &102[mm${}^2$]\\
\hline
3.0 [mm]& 25[mm] & 74[mm${}^2$] & 242[mm]& 52[mm${}^2$] & 0.79$\pi$&
28[mm]   &85[mm${}^2$]\\
\hline
5.0 [mm]& 21[mm]& 104[mm${}^2$]  & 250[mm] & 49[mm${}^2$]& 0.86$\pi$&
23[mm]   &113[mm${}^2$]\\
\hline
  \end{tabular}
\end{table}

From (\ref{eq:ELW}), 
$E(\delta):=\delta\cdot(L-W_\cc)$'s correspond to the elastic energy 
at the critical state, which are similar values, though the width in
photographs in Figure \ref{fig:experiment} 
 is not easy to be determined and must have some errors.
On the other hand, from (\ref{eq:2.14}), it is expected  
that the force $\delta \cdot k_\cc$ is proportional to $\delta\cdot X_\cc $ 
(up to $\alpha$-dependence)
depend on the material properties though we made the notch
in each panel, Table 7 gives the natural results, in which
$\delta \cdot X_\cc$'s are similar values.
The height of $\Lambda$-elastica $X_\Lambda$ is
nearly equal to $X_\cc$ for every $\delta$ and thus
$\delta\cdot X_\Lambda$'s are similar values though we cannot
compare  $\delta\cdot X_\cc$ and $\delta\cdot X_\Lambda$
from mechanical viewpoints because $\alpha$'s in (\ref{eq:2.14})
of both elasitca and $\Lambda$-elastica are irrelevant.

In the experiment, it is expected 
that $\delta\cdot\Delta \cE^{\Lambda_{\phi_0}}$ corresponds 
to the energy of rupture.
After the panel lost the energy  of rupture, $\phi_0$ of 
$\Lambda_{\phi_0}$-elastica is determined
by energy conservation law.

Thus we note that Figure \ref{fig:Belastica_final0} and 
Figure \ref{fig:Belastica_final1} are consistent with 
the experimental results; the larger $W_\cc$ is, the larger $\phi_0$ is.
Our numerical computations also show that the lager $W_\cc$ is,
the lager $\phi_0$ is because the energy gap needs positive.

It means that we provide a novel
investigation of rupture phenomena for the beam bending test.
Further the shape of $\Lambda$-elastica is very interesting
since the shape of $\Lambda$-elastica appears in \cite{MWT} and in \cite{DMJ}.
As mentioned above, we described the transition from elastica to 
$\Lambda$-elastica in the beam bending experiment and $\Lambda$-elastica
mathematically.
We hope that our investigation should have some effects on these 
studies.


\vskip 1.0 cm

\section*{Acknowledgments}
The authors would like to express their sincere gratitude
to  the participants in the
``IMI workshop II: Mathematics of Screw Dislocation'',
September 1--2, 2016, in the
``IMI workshop I: Mathematics in Interface, Dislocation and 
Structure of Crystals'',
August 28--30, 2017, both held in the 
Institute of Mathematics for Industry (IMI),
to  the participants in the
``IMI workshop II: Advanced Mathematical Investigation for Dislocations'',
September 10--11, 2018,
and
``IMI workshop II: Advanced Mathematical Analysis for Dislocation, Interface and Structure in Crystals'',
September 9--10, 2019,
at Kyushu University.
The first author thanks Professor Ryuichi Tarumi for pointing out 
the Weierstrass-Erdmann corner conditions.
This study has been supported by Takahashi Industrial and 
Economic Research Foundation 2018-2019, 08-003-181.

\noindent
Shigeki Matsutani:\\
Faculty of Electrical, Information and Communication Engineering,\\
Kanazawa University\\
Kakuma Kanazawa, 920-1192, JAPAN\\
e-mail: s-matsutani@se.kanazawa-u.ac.jp\\

\noindent
Hiroshi Nishiguchi:\\
National Institute of Technology, Sasebo College,\\
1-1 Okishin-machi, Sasebo, Nagasaki, 857-1193, JAPAN\\

\noindent
Kenji Higashida:\\
National Institute of Technology, Sasebo College,\\
1-1 Okishin-machi, Sasebo, Nagasaki, 857-1193, JAPAN\\

\noindent
Akihiro Nakatani:\\
Department of Adaptive Machine Systems, \\
Graduate School of Engineering, 
Osaka University, 2-1 Yamadaoka, Suita, Osaka 565-0871, JAPAN\\

\noindent
Hiroyasu Hamada:\\
National Institute of Technology, Sasebo College,\\
1-1 Okishin-machi, Sasebo, Nagasaki, 857-1193, JAPAN\\


\bigskip

\begin{thebibliography}{99}


\bibitem{Bigoni}
\By{Bigoni, D}
\Book{Nonlinear Solid Mechanics: Bifurcation Theory and Material Instability}
Cambridge Univ. Press,
2014.
 

\bibitem{Griffiths}
\By{Bryant R. and Griffiths P.}
\Paper{Reduction for Constrained Variational Problems and
$\int \kappa^2/2 ds$}
\Jour{Amer. J.  Math.} \Vol{108} \Yr{1986}\Pages{525-570}.


\bibitem{DMJ}
\By{Dharmadasa Y,  Mallikarachchi H.M.Y.C., and  L\'opez Jim\'enez F.},
\Paper{Characterizing the Mechanics of Fold-lines in Thin Kapton
Membranes}
\Jour{AIAA SciTech Forum}, 8-12 January 2018, Kissimmee, Florida,
10.2514/6.2018-0450.

\bibitem{method}
\By{Euler L.} \Book{Methodus Inveniendi Lineas Curvas Maximi
Minimive Proprietate Gaudentes} 1744, Leonhardi Euleri Opera Omnia
Ser. I vol. 14.

\bibitem{GF}
\By{Gelfand, I. M. and  Fomin, S. V.} 
\Book{Calculus of Variations} 
Dover Publications, 2012.



\bibitem{GP}
\By{Goldstein R. and Petrich D.}
\Paper{The Korteweg-de Vries
Hierarchy as Dynamics of Closed Curves in the Plane}
\Jour{Phys. Rev. Lett.} \Vol{67}  \Yr{1991}\Pages{3203-3206}.


\bibitem{Hagihara}
\By{Hagihara K., Yokotani N. and Umakoshi Y.}
\Paper{Plastic deformation behavior of Mg12YZn with 18R long-period stacking
ordered structure}
Intermetallics \Vol{18} (2010) 267-276.

\bibitem{HessBarrett}
\By{ Hess J. B. and  Barrett C. S.}
\Jour{Trans Am Inst Min Met Eng.}
\Vol{185} \Yr{1949} \Pages{599-606}.

\bibitem{MH}
\By{Mladenov, I. M., Hadzhilazova, M.} 
\Book{The Many Faces of Elastica (Forum for Interdisciplinary Mathematics)}
Springer, 2017.

\bibitem{Mat13}
\By{Matsutani S.}
 \Paper{Euler's elastica and its beyond}
       \Jour{J. Symm. Geom. Phys.} \Vol{17} (2013) \Pages{45-86}.


\bibitem{MWT}
\By{Morigaki Y.,  Wada  H. and  Tanaka Y.}
\Paper{Stretching an Elastic Loop: Crease, Helicoid, and Pop Out}
\Jour{Phys. Rev. Lett.} \Vol{117} \Yr{2016} 198003.

\bibitem{Mum}
\By{Mumford D.}
 \Paper{Elastica and Computer Vision}
  in Algebraic Geometry and its Applications,
      ed. by C.~Bajaj, Springer-Verlag Berlin 1993
     507-518.


\bibitem{Tr}
\By{Truesdell C.}
\Paper{The Influence of
    Elasticity on Analysis:
    The Classic Heritage}
\Jour{Bull. Amer. Math. Soc.} \Vol{9}
    (1983) 293-310.


\bibitem{WW}
\By{Whittaker E. and  Watson G.}
\Book{A Course of Modern Analysis} 4th ed.
\Publ{Cambridge Univ. Press}
\Publaddr{Cambridge}
1927.


\end{thebibliography}
\end{document}